\documentclass[12pt]{amsart}

\makeatletter
\def\l@subsection{\@tocline{2}{0pt}{2.5pc}{5pc}{}}
\def\l@subsubsection{\@tocline{2}{0pt}{5pc}{7.5pc}{}}
\makeatother

\usepackage{fullpage,amssymb}
\usepackage{mathrsfs}
\usepackage{color}
\usepackage{algorithmicx}
\usepackage{algpseudocode}
\usepackage{tikz-cd}
\usepackage{tikz}
\usetikzlibrary{matrix,decorations.pathreplacing}
\usepackage[T1]{fontenc}
\usepackage{amsbsy}
\usepackage{bm}

\usepackage{color}   
\usepackage{hyperref}
\hypersetup{
    linktoc=all,     
    linkcolor=black,  
}

\newtheorem{theorem}{Theorem}[section]
\newtheorem{corollary}[theorem]{Corollary} 
\newtheorem{lemma}[theorem]{Lemma}
\newtheorem{proposition}[theorem]{Proposition}

\theoremstyle{definition}
\newtheorem{definition}[theorem]{Definition}
\newtheorem{remark}[theorem]{Remark}
\newtheorem{example}[theorem]{Example}

\newcommand{\M} {\operatorname{Mat}}
\newcommand{\A} {\mathbb{A}}
\newcommand{\C} {\mathbb{C}}
\newcommand{\Z} {\mathbb{Z}}
\newcommand{\F} {\mathbb{F}}
\newcommand{\N} {\mathbb{N}}
\newcommand{\Q} {\mathbb{Q}}
\newcommand{\R} {\mathbb{R}}
\newcommand{\K}{\mathbb{K}}

\newcommand{\rk} {{\rm rk}}
\newcommand{\Ten} {{\rm Ten}}
\newcommand{\trk} {{\rm trk}}
\newcommand{\brk} {{\rm brk}}
\newcommand{\wrk} {{\rm wrk}}
\newcommand{\Pot} {{\rm Pot}}
\newcommand{\Bpot} {{\rm Bpot}}
\newcommand{\Hom} {{\rm Hom}}
\newcommand{\spa}{{\rm span}}
\newcommand{\MSpec} {{\rm MSpec}}
\newcommand{\Ker} {{\rm Ker}}
\newcommand{\im} {{\rm Im}}
\newcommand{\smdeg}{\text{sm-deg}}

\newcommand{\OO} {\mathcal{O}}
\newcommand{\I} {\mathcal{I}}
\newcommand{\SP} {\mathcal{SP}}
\newcommand{\Pp} {\mathcal{OP}}
\newcommand{\cC} {\mathcal{C}}

\newcommand{\cJ}{\mathcal{J}}

\newcommand{\veps}{\varepsilon}
\newcommand{\tleq}{\unlhd}
\newcommand{\st}{\text{ s.t. }}

\newcommand{\bv}{\textbf{v}}
\newcommand{\bz}{\textbf{z}}
\newcommand{\by}{\textbf{y}}

\newcommand{\ba}{\textbf{a}}
\newcommand{\bb}{\textbf{b}}
\newcommand{\bc}{\textbf{c}}
\newcommand{\be}{\textbf{e}}
\newcommand{\bff}{\textbf{f}}

\newcommand{\bo}{\textbf{0}}

\newcommand{\bp}{\textbf{p}}
\newcommand{\bq}{\textbf{q}}
\newcommand{\br}{\textbf{r}}
\newcommand{\bn} {\textbf{n}}
\newcommand{\bd} {\textbf{d}}

\newcommand{\balpha} {\bm{\alpha}}
\newcommand{\bbeta} {\bm{\beta}}
\newcommand{\bdelta} {\bm{\delta}}
\newcommand{\supp} {{\rm supp}}
\newcommand{\bell}{\ensuremath{\boldsymbol\ell}}

\title{More barriers for rank methods, via a ``numeric to symbolic'' transfer}

\author{Ankit Garg}
\address{Microsoft Research India, Bangalore}
\email{garga@microsoft.com}

\author{Visu Makam} \thanks{Visu's research was supported by NSF grant No. DMS -1638352 and NSF grant No. CCF-1412958.}
\address{School of Mathematics, Institute for Advanced Study, Princeton}
\email{visu@ias.edu}

\author{Rafael Oliveira} 
\address{University of Toronto and Simons Institute for the Theory of Computing}
\email{rafael@cs.toronto.edu}

\author{Avi Wigderson} \thanks{
Avi's research was supported by NST grant No. CCF-1412958.}
\address{School of Mathematics, Institute for Advanced Study, Princeton}
\email{avi@ias.edu}

\begin{document}

\begin{abstract}
We  prove new barrier results in arithmetic complexity theory, showing severe limitations of natural {\em lifting} (aka  {\em escalation}) techniques. For example, we prove that even optimal rank lower bounds on $k$-tensors cannot yield non-trivial lower bounds on the rank of $d$-tensors, for any constant $d>k$. This significantly
extends recent barrier results on the limits of {\em (matrix) rank methods} by~\cite{EGOW18}, which handles the (very important) case $k=2$.

Our generalization requires the development of new technical tools and results in algebraic geometry, which are interesting in their own right and possibly applicable elsewhere. The basic issue they probe is the relation between {\em numeric} and {\em symbolic} rank of tensors, essential in the proofs of previous and current barriers. Our main technical result implies that for every symbolic $k$-tensor (namely one whose entries are polynomials in some set of variables), if the tensor rank is small for every evaluation of the variables, then it is small symbolically. This statement is obvious for $k=2$. 

To prove an analogous statement for $k>2$ we develop a ``numeric to symbolic'' transfer of algebraic 
relations to algebraic
functions, somewhat in the spirit of the implicit function theorem.  It applies in the general setting of inclusion of images of polynomial maps, in the form appearing in Raz's {\em elusive functions} approach to proving {\rm VP $\neq$ VNP}. We give a toy application showing how our transfer theorem may be useful in pursuing this approach to prove arithmetic complexity lower bounds.
\end{abstract}

\maketitle

\tableofcontents

\section{Introduction}\label{sec:intro}

One of the major goals of complexity theory is to prove lower bounds for various models of 
computation. The theory often proceeds in buckets of three steps. The first is to come up with a 
collection of techniques. The second is to be frustrated at the fact that the collection is not 
powerful enough to prove the lower bounds we want. The final step is to prove a `barrier' result on 
the collection of techniques, giving a formal rigorous explanation as to why these techniques do not 
suffice. Then, of course, one searches for new techniques avoiding known barriers, and the process is repeated until (hopefully!) the desired lower bounds are attained.

One common set of techniques, which is ancient but whose prominence and use increases with recent successes (and realization that past methods fit this mold) are {\em lifting} (or {\em escalation}) techniques. Here one aims to derive a lower bound for some strong model, via a reduction to proving a related lower bound on a weaker model (another variant is deriving a strong lower bound from a weak one for the same model). This occurs across computational complexity, in Boolean circuit complexity (e.g.~\cite{Razb90,AlKo10}, arithmetic circuit complexity (e.g.~\cite{NiWi96,HWY10,GKKS13}), proof complexity (e.g. \cite{BPR97,GKRS18}), communication complexity (e.g.~\cite{RaMc99,GPW15,GPW17}) and other computational frameworks (where we only referenced few of many examples).

Here we work in the framework of arithmetic complexity. By far the main technique used in proving lower bounds are the so-called {\em rank methods} (which we will presently call  {\em matrix-rank methods}), which reduce proving lower bounds on numerous arithmetic models and complexity measures to the computations of matrix rank.

The two main complexity measures that we will be studying in this paper are {\em tensor rank} and 
{\em Waring rank} (defined in the subsequent section). For a long time, matrix-rank methods were unable to prove any lower bounds 
that were significantly better than the trivial ones (despite independent work in complexity theory and in algebraic geometry). A sweeping barrier result for this collection was proved 
in~\cite{EGOW18}, explaining why matrix-rank methods will {\em never} deliver better results on these measures! 

In this paper, we focus on extending these barrier results to greater generality against stronger techniques. 
The usual matrix rank is a special case of tensor rank, when we view a matrix rank as a degree-2 tensor (tensors of degree $k$ will be termed $k$-tensors).\footnote{Matrices are naturally equivalent to bi-linear forms, which are degree-2 polynomials. Similarly, tensors  naturally equivalent $d$-linear forms, which are degree-$d$ polynomials. This notation is consistent with Waring rank of homogeneous degree-$d$ polynomials.} Generalized rank methods can be thought of as {\em lifting} tensor rank lower bounds via linear maps to $k$-tensors (polynomials) of degree $k>2$.

Our main result is a barrier to these potentially stronger methods. More precisely, for $d>k>2$, we prove barriers to lifting lower bounds on tensor (Waring) rank for (small) degree $k$-tensors (polynomials) to respective lower bounds for (larger) degree $d$-tensors (polynomials). Indeed, as with matrix-rank methods, even optimal lower bounds on the rank of degree $k$ tensors (polynomials) cannot yield any nontrivial lower bounds for any $d>k$ for any fixed $d$, (up to constant factors). 

Our generalization of~\cite{EGOW18} to lifting from $k>2$ is far from obvious. To overcome 
the difficulties, we will need an influx of new ideas and some algebro-geometric tools. We point out that the 
technical results we prove for this generalization are very natural and general statements, in the spirit of the implicit function theorem, and potentially applicable in other contexts  in mathematics and complexity theory. We also note that while the barriers of~\cite{EGOW18} are valid in fields of arbitrary characteristic, our more general results in this paper only hold for characteristic zero.  

We make two comments of an informal nature, which we believe require further exploration. The first is about the power of the ``weak'' model we are trying to lift. Recall that for $k > 2$, $k$-tensor rank is ${\rm NP}$-hard to compute~\cite{Has90}. So, unlike the barrier result of~\cite{EGOW18}, where the ``simpler model'' is a matrix (namely 2-tensor), and its  rank (which is the lower bound to be lifted) is  computationally easy,  here the lower bound that we are we are assuming, and trying to lift, is itself computationally difficult. Despite that, our new barrier result says even such (possibly hard to prove) lower bounds cannot be lifted to any non-trivial lower bounds in higher degree tensors (or polynomials).  A related second point is that optimal lower bounds on $k$-tensors (for $k$ superconstant and $\leq \log(n)/\log(\log(n))$ where $n$ is the local dimension) {\em can} be lifted to lower bounds on some stronger arithmetic models; Raz~\cite{Raz-elusive} shows how they can imply super-polynomial formula lower bounds! We find that better understanding and reconciling these results is needed.

We stress that the techniques in this paper are very general, and can be applied to get a barrier to lifting result between any two 
sub-additive complexity measures. However, it is not always so easy to predict when the obtained barrier would be non-trivial.

We now proceed to make precise definitions and state the main results. Throughout this paper, our ground field (denoted $\F$) will be an algebraically closed field of characteristic zero. We will restate this assumption again whenever it plays an important role. 

\subsection{Various notions of rank}
 Let $\M_{p,q}$ 
denote the linear space of $p \times q$ matrices with entries in $\F$. The {\em rank} of a matrix $M \in \M_{p,q}$ (over $\F$), has many equivalent definitions. For example, it equals the the dimension of the row span of $M$, as well as the dimension of the  column span of $M$, as well as the size of the  largest non-vanishing minor of $M$. 

The definition of matrix rank we prefer will clarify why it is a ``sub-additive 
complexity measure''. First, note that any rank $1$ matrix of size $p \times q$ is of the form $\ba\bb^t$ 
for some $\ba \in \F^p$ and $\bb \in \F^q$.  Let $S \subseteq \M_{p,q}$ denote the subset of rank $1$ 
matrices; we will call these {\em simple} matrices, and this notion will be used throughout. Also note that the set $S$ of ``simples'' is a spanning set of the linear space $\M_{p,q}$.  The rank of any matrix $M \in \M_{p,q}$ is defined 
as the smallest integer $r$ such that $M = A_1 + A_2 + \dots + A_r$ for some 
$A_1,A_2,\dots, A_r \in S$. This definition of matrix rank is 
equivalent to any of the definitions above. What is nice about it is that it motivates the following 
vast generalization. 

\begin{definition} [$S$-rank]
Let $V$ be a vector space, and $S \subseteq V$ be a {\em spanning} subset. For $v \in V$, we define its $S$-rank $\rk_S(v)$ 
to be the smallest integer $r$ such that $v =  s_1 + s_2 + \dots + s_r$ for some 
$s_1, s_2,\dots, s_r \in S$.
\end{definition}

We want to think of $S$ as a set of {\it simple} elements, and $\rk_S(v)$ as the sub-additive complexity of $v$ with respect to this set of simples $S$. Both tensor rank and Waring rank will be special cases of $S$-rank for particular choices of $S$ in vector spaces $V$.

We define
$$
\Ten(n,d) := \underbrace{\F^n \otimes \F^n \otimes \dots \otimes \F^n}_{d},
$$
the space of degree-$d$ tensors with (local) dimension\footnote{One can easily extend the definition to tensors with different local dimensions in each coordinate, as when moving from square to rectangular matrices.} $n$. A tensor which is a product of linear forms, namely of the form 
$\bv_1 \otimes \bv_2 \otimes \dots \otimes \bv_d \in \Ten(n,d)$, is called a {\it simple tensor} or a  
{\it rank $1$ tensor}.

\begin{definition} [Tensor rank]
Let $S := \{\bv_1 \otimes \bv_2 \otimes \dots \otimes \bv_d\ |\ \bv_i \in \F^n \ \forall i\} 
\subseteq \Ten(n,d)$. For a tensor $T \in \Ten(n,d)$, we define its tensor rank 
$\trk(T) \triangleq \rk_S(T)$.
\end{definition}


\begin{example}\label{2ten-mat}
There is a natural identification $\F^p \otimes \F^q = \M_{p,q}$ as follows. Let 
$\{\be_i\}_{1 \leq i \leq p}$ and $\{\bff_j\}_{1 \leq j \leq q}$ denote the standard basis for $\F^p$ and 
$\F^q$ respectively. Then $\{\be_i \otimes \bff_j\}$ is a basis for $\F^p \otimes \F^q$. We identify 
$\be_i \otimes \bff_j$ with the elementary matrix $E_{i,j}$ that has an $1$ in the $(i,j)^{th}$ spot and 
$0$'s everywhere else. 

A concise description of the isomorphism is given by 
$\sum_i \ba_i \otimes \bb_i \mapsto \sum_i \ba_i\bb_i^t$, 
where $\ba_i \in \F^p$ and $\bb_i \in \F^q$. This elucidates the fact that under this identification, 
tensor rank goes to matrix rank.
\end{example}

Let $P(n) := \F[x_1,\dots,x_n]$ denote the polynomial ring in $n$ variables. This has a natural grading 
given by (total) degree. In other words, we have $P(n) = \oplus_{d = 0}^{\infty} P(n,d)$, where $P(n,d)$ 
denotes the homogeneous polynomials of degree $d$. Waring rank is $S$-rank, where the set of simples 
$S$ will be the subset consisting of $d^{th}$ powers of linear forms.

\begin{definition} [Waring rank]
Let $S := \{ \ell^d\ |\  \ell \in P(n,1)\} \subseteq P(n,d)$. For a degree $d$ homogeneous polynomial 
$f \in P(n,d)$, we define its Waring rank $\wrk(f) \triangleq \rk_S(f)$. 
\end{definition}

\begin{example} \label{glynn-formula}
Suppose $d < n$, and consider the monomial $x_1\cdot x_2 \cdots x_d \in P(n,d)$. We can write this as a sum of $2^{d-1}$ powers of linear forms (see \cite{Glynn}):
$$
x_1\cdot x_2 \cdots x_d = \frac{1}{2^{d-1}} \sum_{(\delta_2,\dots,\delta_d) \in \{1,-1\}^{d-1}} (-1)^{\delta_2 + \delta_3 + \dots + \delta_d} (x_1 + \delta_2 x_2 +  \dots + \delta_d x_d)^d.
$$
This means in particular that $\wrk(x_1\cdot x_2 \cdots x_d) \leq 2^{d-1}$. But in fact (see \cite{CCG11,BBT12}), it is an equality! 
\end{example}


\subsection{Sub-additive measures}

A natural approach to prove lower bounds on rank is to use sub-additive measures. For this section, 
let $S$ be a spanning subset of a vector space $V$. 

\begin{definition} [Sub-additive measure]
A sub-additive measure for $S$-rank is a function $\mu: V \rightarrow \R_{\geq 0}$ such that 
$\mu(v+w) \leq \mu(v) + \mu(w)$ for all $v,w \in V$. For any subset $T \subseteq V$, we define 
$\mu(T) = \max\{\mu(v) \ | \ v \in T\}$. 
\end{definition}
The simple way in which this is used to prove lower bounds is the inequality $\rk_S(v) \geq \mu(v)/\mu(S)$.  

Observe that $\rk_S$ is itself a sub-additive measure, but it is difficult to compute. So, one would 
like to use a different sub-additive measure which is simpler to compute. Indeed, the last two sentences of course apply to many other sub-additive complexity measures, e.g. various forms of circuit and proof complexity. Many important  lower bounds in 
arithmetic complexity are obtained in this fashion, such as the partial derivatives method introduced in 
Computer Science by~\cite{N91, NW96}, its generalization, the shifted partial derivatives method -
introduced by~\cite{K12} and developed further in~\cite{GKKS14, KS17}.


Every sub-additive measure will give some lower bounds, but the important question is whether these 
will be strong enough. From our observations above, the best possible lower bound that a 
sub-additive measure $\mu$ can give on any element $v\in V$ (explicit or non-explicit) is $\mu(V)/\mu(S)$. We will define this barrier as the potency 
of the sub-additive measure $\mu$.


\begin{definition} [Potency]
For $V,S$ as above, and any  sub-additive measure $\mu: V \rightarrow \R_{> 0}$, define its \emph{potency} as
$$
\Pot(\mu) \triangleq \mu(V)/\mu(S).
$$
\end{definition}

In short: \emph{strong lower bounds require a potent 
sub-additive measure.} Typically in existing lower bounds, such measures are (intuitively or computationally) easy to compute (like matrix rank).

\subsection{Matrix-rank methods}\label{Matrix-rank-methods}

Due to the focus of this paper, we deviate in notation from our precursor barrier paper~\cite{EGOW18} and from many arithmetic lower bound papers, calling {\em matrix-rank methods} what they all call {\em rank methods}. This highlights the fact that in all these previous papers, the only rank methods used were based on matrix rank, whereas here we extend this to study the power of using rank of higher degree tensors and polynomials to prove new lower bounds.
 
Matrix-rank methods are a large collection of sub-additive measures that are simple to compute. For tensor rank and Waring rank, numerous known lower bounds fall under the purview of matrix-rank methods, see for example \cite{LO,Landsberg-nt, DM, DM2, IK99, Kanev, GL17, LM08, LT10, Farnsworth}.

\begin{definition} [Matrix-rank method]
Let $V$ be a vector space and let $S \subseteq V$ be a spanning subset. Any {\em linear} map 
$\phi: V \rightarrow \M_{p,q}$ is called a matrix-rank method. The complexity measure
associated with $\phi$ is given by $\mu_\phi : V \rightarrow \R_{\geq 0}$ where
$\mu_\phi(v) := \rk(\phi(v)).$
\end{definition}

From the definition above and the properties of matrix rank, one sees immediately that 
$\mu_{\phi}$ is a sub-additive measure.
If we let $\mu_\phi(S) = \max\{\rk(\phi(s)) \mid s \in S\}$ as above, then 
for all $v \in V$ we can get a lower bound: 
$$
\rk_S(v) \geq \frac{\rk(\phi(v))}{\mu_\phi(S)}.
$$

We will often obfuscate the matrix-rank method $\phi$ with the corresponding sub-additive measure 
$\mu_\phi$. In particular, we will call $\Pot(\phi) := \Pot(\mu_\phi)$ the potency of the 
matrix-rank method $\phi$.

\begin{example}[Trivial matrix-rank method]\label{flattening}
We discuss the most basic, naive example of a matrix-rank method that can be used to prove lower bounds 
for tensor rank. By grouping the different tensor factors into two groups (sometimes called {\em flattening}, and can be pictured as such), one can view a tensor in 
$$
\Ten(n,d) = (\underbrace{\F^n \otimes \F^n \otimes \dots \otimes \F^n}_{p}) 
\otimes (\underbrace{\F^n \otimes \F^n \otimes \dots \otimes \F^n}_{q})
$$
as a tensor in $\F^{n^p} \otimes \F^{n^q}$. The latter can be interpreted as an $n^p \times n^q$ matrix as in Example~\ref{2ten-mat}. This gives a linear map $\phi: \Ten(n, d) \rightarrow \M_{n^p, n^q}$ i.e., a matrix-rank method. Let $S$ be the set of simple (or rank $1$) tensors. We observe that $\mu_\phi(S) = 1$, and $\mu_\phi(\Ten(n,d)) = \min\{n^p,n^q\}$, the largest possible rank of an $n^p \times n^q$ matrix. Thus $\Pot(\phi) = \min\{n^p,n^q\}$. So, the potency of these `obvious' matrix-rank methods is at most $n^{\lfloor d/2 \rfloor}$, which is attained when we take $p =\lfloor d/2 \rfloor$. 
\end{example}
 
By a simple dimension count, one can show that most tensors in $\Ten(n,d)$ have tensor rank at least $\frac{n^{d-1}}{d}$. This is much larger than the potency of the obvious rank methods in the previous example (for  fixed $d$ and large $n$). A line of research that was pursued for over a decade with little success was to find more potent matrix-rank methods. While such methods with larger potency have been found (often quite sophisticated with algebraic-geometric ideas), the improvement they yield was very modest -- only by small constant factors. For example, for $3$-tensors, the best known improvement is only by a constant factor of $2$, see \cite{Landsberg-nt, DM,DM2}\footnote{One can obtain a larger constant factor of $3$ using techniques that do not fall under rank methods, see \cite{AFT11}.}.

 Eventually this state of affairs was explained by the barrier result of~\cite{EGOW18}; {\em no} matrix-rank method can do much better than the naive flattening.

\begin{theorem} [\cite{EGOW18}] \label{barr-egow-ten}
For any matrix-rank method $\phi: \Ten(n,d) \rightarrow \M_{k,l}$, its potency 
$$
\Pot(\phi) \leq 2^d n^{\lfloor d/2 \rfloor}.
$$
\end{theorem}

In particular, for $d$ constant, odd integer, matrix-rank methods can only prove a lower bound of the form $\Omega(n^{(d-1)/2})$, while most tensors have a quadratically larger tensor rank $\Omega(n^{d-1})$. 

A similar barrier result for proving lower bounds on Waring rank by matrix-rank methods was also proved in \cite{EGOW18}.
\begin{theorem} [\cite{EGOW18}] \label{barr-egow-war}
For any matrix-rank method $\phi: P(n,d) \rightarrow \M_{k,l}$, its potency
$$
\Pot(\phi) \leq Y_{n,d} + Z_{n,d}
$$
where $Y_{n,d} = {n + \lfloor d/2 \rfloor \choose n}$ is the number of monomials of degree 
$\leq \lfloor d/2 \rfloor$ in $n$ variables, and $Z_{n,d}$ is the number of monomials of degree $\leq d - (\lfloor\frac{d}{2}\rfloor +1)$ in $n$ variables.
\end{theorem}

Again, this result matches the ``trivial'' lower bounds on Waring rank, and is quadratically away from the Waring rank of most polynomials in $P(n,d)$.

\subsection{Generalized rank methods and statements of main results}
The purpose of this paper is to investigate the potency of a larger class of sub-additive measures, and prove barriers for these. We consider two types of generalized rank methods. Let $d>k\geq2$. Matrix-rank methods lift degree-2 lower bounds to degree $d$ lower bounds (for tensors and polynomials). Now we lift degree-$k$ tensor and  Waring rank lower bounds to degree $d$ ones for $d > k$. Again, it is best to  think of $d,k$ as constants (although our results are for all values), and $n$ going to infinity as the main complexity parameter. Repeating a comment made earlier, note that now the assumed lower bounds (for $k>2$) that we are trying to lift are {\em not} easy to compute (in contrast to  matrix-rank methods where $k=2$).

Summarizing this section, our main results naturally extend the ones in~\cite{EGOW18}. First, for these more general methods, there is a trivial way to use them, analogous to flattening in matrix-rank methods, which for every $k<d$ give much weaker bounds than the tensor/Waring rank for most degree $d$ tensors/polynomials.  Second, our barriers show that {\em any} use of these general methods (despite lifting  hard-to-prove lower bounds) {\em cannot} improve their trivial use by more than a constant factor (for constant $d$). Third, the proofs of our barrier results also follow the general strategy of~\cite{EGOW18}. However, the case $k>2$ seems to raise major, interesting difficulties in implementing that strategy, which require new ideas, as well as more sophisticated tools from algebraic geometry. These in turn lead us to prove purely algebraic results regarding polynomial maps which we believe can be useful way beyond the context of this paper, both in algebraic complexity and in algebraic geometry. We will encapsulate this main result in the next subsection as well, and discuss at length the difficulties, ideas and tools in Section~\ref{sec:new ideas}. 

We now turn to formally define the generalized rank methods we consider, and state our main results. 

\begin{definition} [$T_k$-rank method]
Let $V$ be a vector space and let $S \subseteq V$ be a spanning subset. A linear map $\phi: V \rightarrow \Ten(m,k)$ is called a $T_k$-rank method. Thus, matrix-rank methods are simply $T_2$-rank methods. The function $\mu_\phi$ defined by $\mu_\phi(v) = \trk(\phi(v))$ for $v \in V$ is a sub-additive measure, and $\Pot(\phi) = \mu_\phi (V)/\mu_\phi (S)$.
\end{definition}

\begin{example} [Trivial $T_k$-rank method] \label{triv.t_k}
Consider $\Ten(n,d)$, where $d = rk$ for simplicity. In the spirit of simple flattenings of Example~\ref{flattening},  by clubbing together the tensor factors into $k$ groups of size $r$, we get a linear map $\Ten(n,d) \rightarrow \Ten(n^r,k)$. This "trivial" $T_k$-rank method has potency $\Omega((n^r)^{k-1})$. To give a frame of reference for the theorem below, we note that $n^{r(k-1)} = n^{\lfloor\frac{(k-1)d}{k}\rfloor}$. 
\end{example}

Recall that we assume throughout the paper that the ground field $\F$ is algebraically closed and characteristic zero. This is important in our main results (i.e., Theorems~\ref{TtoT},~\ref{PtoT},~\ref{TtoP} and ~\ref{PtoP}), so we will restate this assumption.

\begin{theorem} \label{TtoT}
Suppose that the ground field is algebraically closed and characteristic zero. For any $T_k$-rank method $\phi:\Ten(n,d) \rightarrow \Ten(m,k)$, its potency
$$
\Pot(\phi) \leq A_{d,k} \cdot (n^{\lfloor\frac{(k-1)d}{k}\rfloor}),
$$
where $A_{d,k} = k^d$.
\end{theorem}

The theorem holds for all values of $k,d,n,m$! Let us say a few words on these parameters.
First note that it recovers (with $k=2$) Theorem~\ref{barr-egow-ten} of \cite{EGOW18}. Next note that, as in~\cite{EGOW18},  for constant $d$ the upper bound is a constant factor away from the trivial use of the method. Finally, note that, again as in~\cite{EGOW18}, our theorem holds for any value ($m$ here) of the dimension of the image space, and it does not assume anything (in particular explicitness) of linear map $\phi$ used by the method!


\begin{remark}
The above theorem is especially interesting in the case for $\Ten(n,4)$. The trivial lower bound, the barrier for matrix-rank methods and the barrier for $T_3$-rank methods are all quadratic in $n$, differing only in a constant factor. Hence, even if one had access to an oracle for tensor rank of $3$-tensors, one could still not prove super-quadratic lower bounds for the tensor rank of tensors in $\Ten(n,4)$. 
\end{remark}

The following table puts the degree of tensors against lower bounds obtainable by different classes of rank methods. We suppress constant terms.

\begin{center}
$$
\begin{array} {|c|c|c|c|c|c|}
\hline
& \text{Trivial} & \text{Rank methods} & \text{Trivial } T_3 & \text{Best } T_3 & \text{Desired} \\
\hline
3-\text{tensors} & n & n & n^2 & n^2 & n^2 \\
\hline
4-\text{tensors} & n^2 & n^2 & n^2 & n^2 & n^3 \\
\hline
5-\text{tensors} & n^2 & n^2 & n^3 & n^3 & n^4\\
\hline
\end{array}
$$
\end{center}


The following three results are in the same spirit.

\begin{theorem} \label{PtoT}
Suppose that the ground field is algebraically closed and characteristic zero. For any $T_k$-rank method $\phi:P(n,d) \rightarrow \Ten(m,k)$, its potency
$$
\Pot(\phi) \leq B_{d,k} \cdot (n^{\lfloor\frac{(k-1)d}{k}\rfloor})
$$
for some constant $B_{d,k}$ depending only on $d$ and $k$.
\end{theorem}

One can compute an explicit upper bound for the constant $B_{d,k}$, but it is quite messy. If the reader is so inclined, they may extract an explicit upper bound from Corollary~\ref{actual-war-bound}.

\begin{definition} [$W_k$-rank method]
Let $V$ be a vector space and let $S \subseteq V$ be a spanning subset. A linear map $\phi: V \rightarrow P(m,k)$ is called a $W_k$-rank method. The function $\mu_\phi$ defined by $\mu_\phi(v) = \wrk(\phi(v))$ for $v \in V$ is a sub-additive measure, and $\Pot(\phi) = \mu_\phi (V)/\mu_\phi (S)$.
\end{definition}

\begin{theorem} \label{TtoP}
Suppose that the ground field is algebraically closed and characteristic zero. For any $W_k$-rank method $\phi:\Ten(n,d) \rightarrow P(m,k)$, its potency
$$
\Pot(\phi) \leq C_{d,k} \cdot (n^{\lfloor\frac{(k-1)d}{k}\rfloor}),
$$
where $C_{d,k} =2^{k-1} k^d$.
\end{theorem}

\begin{theorem} \label{PtoP}
Suppose that the ground field is algebraically closed and characteristic zero. For any $W_k$-rank method $\phi:P(n,d) \rightarrow P(m,k)$, its potency
$$
\Pot(\phi) \leq 2^{k-1} B_{d,k} \cdot (n^{\lfloor\frac{(k-1)d}{k}\rfloor}),
$$
where $B_{d,k}$ is the same constant as in Theorem~\ref{PtoT}.
\end{theorem}


\begin{remark}
For matrix-rank methods, there is an alternate approach to proving barriers for the potency using the notion of cactus rank (see Section~\ref{sec:gen-cac-imp}). This approach doesn't seem to have an obvious generalization to $T_k$ and $W_k$-rank methods. It would be interesting to understand if there is an appropriate generalization that would also lead to the same barriers for $T_k-$rank methods and $W_k-$rank methods that we obtain in this paper.
\end{remark}

\subsection{Numeric to symbolic transfer}
The key new ingredient in this paper is a very general ``numeric to symbolic transfer'' statement. We will first state the theorem and then explain its meaning. 

\begin{theorem} \label{main-ag}
Let $\F$ be an algebraically closed field of characteristic zero. Suppose $L: \F^n \rightarrow \F^m$ is a polynomial map, and $M: \F^r \rightarrow \F^m$ is another polynomial map such that $\im(L) \subseteq \im(M)$. Let $\bz = (z_1,\dots,z_n)$ be a vector of indeterminates. Then there exists $ \bc = (c_1,\dots,c_n) \in \F^n$ such that
$$
L(\bz + \bc) = M(p_1(\bz),\dots,p_r(\bz)),
$$
where $p_1(\bz),\dots,p_r(\bz)$ are ($n$-variate) power series around $\bo$.
\end{theorem}

That the map $L = (L_1,\dots,L_m)$ is a polynomial map simply means that each component $L_i:\F^n \rightarrow \F$ is a polynomial function (similarly for $M$). For an exponent vector $\be = (e_1,\dots,e_n) \in \N^n$, we use the shorthand $\bz^{\be} = z_1^{e_1}z_2^{e_2}\dots z_n^{e_n}$ to denote monomials. By an $n$-variate power series around $\bo$, we mean an expression of the form $q(\bz) = \sum_{\be \in \N^n} q_\be \bz^\be$, where $q_\be \in \F$. Addition and multiplication on power series is defined in the standard way, and so it makes sense to plug in a vector of power series into $M$. Equality of power series is purely symbolic\footnote{When $\F= \C$, one can interpret these power series as analytic functions (provided they converge in some neighbourhood), and then equality of power series is the same as equality of functions. For other fields, there is not always a reasonable interpretation of power series as functions.} -- two power series $q(\bz)$ and $p(\bz)$ are said to be equal if $q_{\be} = p_{\be}$ for all $\be \in \N^n$.

The hypothesis $\im(L) \subseteq \im(M)$ is a condition on the numeric evaluations of $L$ and $M$. To interpret the conclusion, first observe that the right hand side is (apriori) a vector of power series. The left hand side is a vector of polynomials, and polynomials are power series. Thus the conclusion is an equality as vectors of power series, which is a symbolic statement -- hence the interpretation of the above theorem as a ``numeric to symbolic transfer'' statement. 

\begin{remark}
The above result is very much in the spirit of the implicit function theorem and the constant rank theorem. However, it does not seem to be a straightforward consequence. If this were the case, we should expect a similar statement for smooth functions -- if we take $\F = \R$, $L,M$ to be $C^{\infty}$ functions, the $p_i(\bz)$ to be $C^{\infty}$ functions on some small neighbourhood of $0$, and ask for the conclusion to be an equality (as functions) on a small neighbourhood. No such statement seems to be known to the best of our knowledge.\footnote{Experts have suggested that it is likely false in this setting. However, constructing an explicit counterexample seems to be difficult.}
\end{remark}


Our use of Theorem~\ref{main-ag} will be in the context:

\begin{corollary} \label{cor:ag}
Let $\F$ be an algebraically closed field of characteristic zero. Let $L: \F^n \rightarrow \Ten(m,k)$ be a polynomial map. Let $\bz = (z_1,\dots,z_n)$ be a vector of indeterminates. If $\trk(L(\bbeta)) \leq a$ for all $\bbeta \in \F^n$, then there exists $\bc \in \F^n$ such that we have a power series decomposition
$$
L(\bz + \bc) = \sum_{i=1}^a \bp_i^{(1)}(\bz) \otimes \bp_i^{(2)}(\bz) \otimes \dots \otimes \bp_i^{(k)}(\bz),
$$
where $\bp_i^{(j)}(\bz)$ is an $m$-dimensional vector of power series in the variables $\bz$ around $\bo$.
\end{corollary}

The above corollary is an instantiation of the above theorem, and we defer the details to Section~\ref{Sec:gen2symb}. The barriers for matrix-rank methods in \cite{EGOW18} also required the case $k = 2$ in the above corollary. However, that special case is straightforward to prove, which we will see in Section~\ref{sec:new ideas}. With the exception of the above result, the rest of the arguments for proving our main results are natural generalizations of the arguments for the $k = 2$ case in \cite{EGOW18}.

Theorem~\ref{main-ag} requires non-trivial notions and results from algebraic geometry to prove. On the other hand, the statement itself is accessible and neat, and we speculate that it will find more uses in complexity theory. One possible use is to prove lower bounds via elusive functions (see \cite{Raz-elusive}).  The elusiveness of a function is a numerical condition and fits precisely into the setup of the above theorem. Thus, Theorem~\ref{main-ag} allows us for a symbolic interpretation of this condition. The advantage is that this brings new tools to the table, which we demonstrate in toy cases (see Section~\ref{sec:elusive}).

\subsection{Results on border rank, set-multihomogeneous rank and cactus rank}
We give a brief overview of some additional results that we include in this paper, and we defer the details to the appropriate sections.

\begin{enumerate}
\item {\bf Border rank:} For simplicity, we have ignored the notion of border rank in the introduction so far. In Section~\ref{sec:border-rank}, we prove barriers to lifting border rank of tensors/polynomials. Incorporating the notion of border rank is not straightforward, and again requires results from algebraic geometry. \\

\item  {\bf Matching barriers obtained from cactus rank:} For matrix-rank methods, the current barriers in~\cite{EGOW18} do not match precisely the barriers obtained by cactus rank arguments (the gap is quite small). By introducing an additional idea, we match the barriers obtained in both approaches in Section~\ref{sec:gen-cac-imp}.

\item {\bf Set-multihomogeneous rank:} We discuss barriers for matrix-rank methods for set-multihomogeneous rank (a generalization of tensor and Waring ranks) in Section~\ref{sec:gen-cac-imp}. \\

\end{enumerate}

\subsection{Organization}
In Section~\ref{sec:prelim}, we collect some notation. In Section~\ref{sec:new ideas}, we give a proof sketch of the barriers for matrix rank methods (in \cite{EGOW18}). We also discuss the issues with generalizing the arguments and the new ideas to overcome these; in particular we discuss the ingredients of the numeric to symbolic transfer statement (Theorem~\ref{main-ag}). Section~\ref{Sec:gen2symb} contains a brief introduction of notions in algebraic geometry we require, and a detailed proof of Theorem~\ref{main-ag}. The barriers to tensor (resp. Waring) rank lower bound methods are established in Section~\ref{sec:tensor-rank} (resp. Section~\ref{sec:war-rank}), thereby proving our main results.

We study border rank methods and establish barriers for these in Section~\ref{sec:border-rank}. This  requires a careful interplay between algebraic and topological border rank. In Section~\ref{sec:gen-cac-imp}, we discuss barriers for set multi-homogenous rank, as well as match the barriers obtained from our techniques with the barriers coming from cactus rank. Finally, in Section~\ref{sec:elusive}, we discuss elusive functions and their importance in lower bounds, and suggest a symbolic approach using our numeric to symbolic transfer statement.

\section{Notation}\label{sec:prelim}

In this section we establish additional notation to the ones given in the previous section, 
and state basic facts which will be used throughout the paper.

For a ring $R$, we define $\Ten_R(n,d) := (R^n)^{\otimes d}$ as the module defined by the set of
degree $d$ tensors with local dimension $n$ and entries given by elements of the ring $R$. When the
ring is clear from the context, we omit it from the definition, as we did in the previous section.

We will denote elements of a field or of a ring with lowercase normal or greek letters, such as 
$a, b, c, \alpha, \beta, \gamma$.
Given a vector space (or an $R$-module) $V$, such as $\F^n$, we will denote elements of this vector
space with boldface letter, for instance $\bv \in V$. Similarly, we will also denote a set (or vector) of indeterminates with 
boldface letters $\bz := (z_1, \ldots, z_m)$. We will sometimes think of $\bz$ as a set and sometimes as a vector and this will be obvious from the context. For example, we think of it as a set when we write the function field $\F(\bz) = \F(z_1,\dots,z_m)$, and we think of it as a vector when we write $L(\bz)$ for some function $L$ that takes $n$ inputs (as we do in Theorem~\ref{main-ag}).


We will use the following shorthand notation to refer to a monomial: $\bz^\be = \prod_{i=1}^m z_i^{e_i}$,
where $\be \in \N^m$. Given a polynomial $f(\bz) \in P(m)$, we will denote its degree by $\deg(f)$.
Thus, the degree of the monomial $\bz^\be$ is given by $\deg(\bz^\be) = e_1 + \cdots + e_m$. We will also write $\deg(\be)$ for $\deg(\bz^\be)$ as it simplifies notation.

A power series in $\bz$ around $\bc = (c_1,\dots,c_m) \in \F^m$ is an expression of the form $p(\bz) = \sum_{\be \in \N^m} p_{\be} (\bz-\bc)^\be$. Note that $(\bz-\bc)^\be = \prod (z_i - c_i)^{e_i}$. Given two power series $p(\bz)$ and $q(\bz)$, we can add or multiply them in the obvious fashion. This gives the collection of all power series in $\bz$ around $\bc$ the structure of a ring, which we call the ring of power series.

\begin{definition}[Ring of Power Series]
We denote by $\F[|\bz -\bc|] = \F[|z_1 - c_1,\dots,z_m-c_m|]$ the ring of power series in $\bz$ around $\bc$.
\end{definition}

\section{Proof strategies for previous results and new ideas} \label{sec:new ideas}
 The high-level strategy for proving our main results is similar to the barriers for matrix-rank methods in \cite{EGOW18}. Hence, we will give a sketch of the arguments in \cite{EGOW18}, which will help us identify the difficulties in generalization, and the new ideas (primarily the numeric to symbolic transfer) that are required to overcome this.

\subsection{Barriers for matrix-rank methods: Proof sketch} \label{Sec:pfsketch}
The following observation is key on how potential upper bounds, and thus barriers are obtained -- any matrix of the form
    

\begin{center}
\includegraphics[scale=0.5]{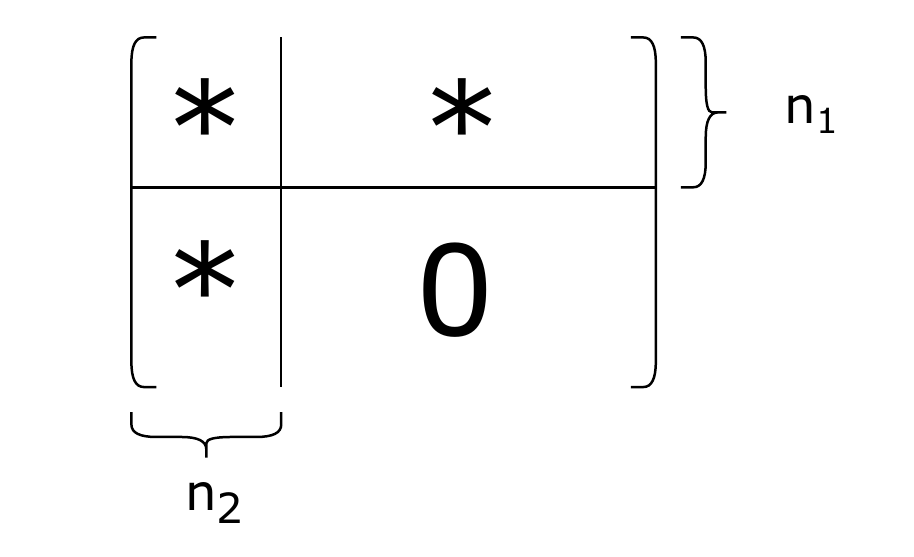}
\end{center}

\noindent has rank at most $n_1 + n_2$. Let us also note that such matrices form a closed set under addition, and in fact a linear subspace. So, the sum of any number of such matrices will also be of this form, and hence have rank at most $n_1 + n_2$. In particular, suppose $\phi: V \rightarrow \M_{k,l}$ is a matrix-rank method (for some spanning set of simples $S\subseteq V$). Further, (under a suitable basis change) suppose that for every $s \in S$, $\phi(s)$ is of the form above. Now, for any $v \in V$, we can write $v = s_1 + s_2 + \dots + s_r$ for some $s_i \in S$. Thus one sees immediately that $\phi(v) = \phi(s_1) + \phi(s_2) + \dots + \phi(s_r)$ is also of the same form (because matrices of such form are closed under addition), and hence has rank $\leq n_1+ n_2$. So, $\mu_\phi(V) = \max\{\rk(\phi(v))\ |\ v \in V\} \leq n_1 + n_2$. Thus, this gives an upper bound on potency  
$$
\Pot(\phi) = \frac{\mu_\phi(V)}{\mu_\phi(S)} \leq \frac{n_1 + n_2}{\mu_\phi(S)}.
$$

Let us identify $\M_{k,l}$ with $\F^k \otimes \F^l$ in the natural fashion. The condition discussed above can be rephrased as having subspaces $U_1 \subseteq \F^k$ and $U_2 \subseteq \F^l$ (with $\dim U_1 = n_1$ and $\dim U_2 = n_2$) such that for all $s \in S$, $\phi(s) \in U_1 \otimes \F^l + \F^k \otimes U_2$. We record this as a lemma for further use.

\begin{lemma} \label{basic}
Let $\phi: V \rightarrow \M_{k,l}$ be a matrix-rank method such that $a = \max\{\rk(s)\ |\ s \in S\} = \mu_\phi(S)$. Suppose we have subspaces $U_1 \subseteq \F^k$ and $U_2 \subseteq \F^l$ such that for all $s \in S$, $\phi(s) \subseteq U_1 \otimes \F^l + \F^k \otimes U_2$. Then, we have $\phi(V) \subseteq U_1 \otimes \F^l + \F^k \otimes U_2$, and consequently,
$$
\Pot(\phi) = \frac{\mu_\phi(V)}{\mu_\phi(S)} \leq \frac{\dim(U_1) + \dim(U_2)}{a}.
$$

\end{lemma}

\begin{definition} [Parametrization]
An (algebraic) parametrization of a spanning subset $S \subseteq V$ is simply a polynomial map $\psi: \F^n \rightarrow V$ such that $\psi(\F^n) = S$. 
\end{definition}

The ability to parametrize simples is crucial in the proofs of barriers. In the setting of Waring rank, i.e, $V = P(n,d)$ and 
$S = \{\ell^d\ | \ \ell \in P(n,1)\}$, we have a parametrization $\psi: \F^n \rightarrow P(n,d)$ given by $(\alpha_1,\dots,\alpha_n) \mapsto (\sum_j \alpha_j x_j)^d$, which is a homogenous polynomial map of degree $d$ (by which we mean that when written in coordinates, it is given by homogenous polynomials of degree $d$). The proof of the barrier seems to depend {\em only} on the nature of the parametrization (that it requires $n$ variables and is homogenous of degree $d$). We now proceed to give a proof sketch of Theorem~\ref{barr-egow-war}.

\begin{proof} [Proof sketch of Theorem~\ref{barr-egow-war}]
Let $V = P(n,d)$ and $S = \{\ell^d\ | \ \ell \in P(n,1)\}$, and $\psi:\F^n \rightarrow V$ be the parametrization of $S$ described above. Composing the matrix-rank method $\phi: V \rightarrow \M_{k,l}$ with the parametrization $\psi$ gives a map $L := \phi \circ \psi : \F^n \rightarrow \M_{k,l}$.

\begin{enumerate}
\item {\bf The starting point:}
It is easy to see that $\im(L) = \phi(S)$. Thus, the map $L$ has the property that $\rk(L(\beta_1,\dots, \beta_n)) \leq  a = \mu_\phi(S)$ for all $\bbeta = (\beta_1,\dots, \beta_n) \in \F^n$.  \\

\item {\bf A symbolic decomposition in the function field:} 
Let $\bz = (z_1,\dots,z_n)$ be a vector of indeterminates. The above statement implies the symbolic statement that $\rk (L(z_1,\dots,z_n)) \leq a$. Note here that $L(z_1,\dots,z_n)$ is a $k \times l$ matrix with entries in the polynomial ring $\F[z_1,\dots,z_n]$ (and hence in the function field $\K = \F(z_1,\dots,z_n)$). So, we take the rank of the matrix over this function field $\K$. So, for some $\bp_i(\bz) \in \K^k$ and $\bq_i(\bz) \in \K^l$, we can write 
$$
L(z_1,\dots,z_n) = \sum_{i=1}^a \bp_i(\bz) \otimes \bq_i (\bz).
$$ \\ 

\item {\bf A power series decomposition:}
 Each $\bp_i(\bz) = (p_{i1}(\bz),\dots,p_{ik}(\bz))$ and $\bq_i(\bz) = (q_{i1}(\bz),\dots,q_{il}(\bz))$, where $p_{ij}(\bz),q_{ij}(\bz) \in \K$ are rational functions. Rational functions have power series expansions wherever they are defined (i.e., where the denominator doesn't vanish). Since, we have finitely many rational functions, we can choose a $\bc \in \F^n$ such that after the shift $\bz \mapsto \bz + \bc$, they are all defined at $\bo \in \F^n$. For such a $\bc$, we have power series expansions around $0$:
$$
\bp_i(\bz + \bc) = \sum_{\be \in \N^n} \bp_{i,\be} \bz^{\be} \text{ and } \bq_i(\bz + \bc) = \sum_{\bff \in \N^n} \bq_{i,\bff} \bz^{\bff},
$$
where $\bp_{i,\be} \in \F^k$ and $\bq_{i,\bff} \in \F^l$, and $\bz^{\be} = \prod_{i=1}^n z_i^{e_i}$. We refer to the $\bp_{i,\be}$'s and $\bq_{i,\bff}$'s as coefficient vectors. This gives the power series decomposition:
 
$$
L(\bz + \bc) = \sum_{i=1}^a \left(\sum_{\be} \bp_{i,\be} \bz^{\be}\right) \otimes \left(\sum_{\bff} \bq_{i,\bff} \bz^{\bff} \right).
$$

\item {\bf A finite monomial decomposition using degree of parametrization:}


Since $\phi$ is a linear map, and $\psi$ is homogenous of degree $d$, the composition $L$ is also homogenous of degree $d$, i.e., $L(\bz)$ is a matrix consisting of homogenous polynomials of degree $d$. So, $L(\bz+\bc)$ is a matrix consisting of polynomials of degree $\leq d$. Hence, if we set $I = \{ (i,\be,\bff)\ |\ \deg(\be) + \deg(\bff) \leq d, 1 \leq i \leq a\}$\footnote{Recall that we use $\deg(\be)$ as shorthand for the degree of the monomial $\bz^{\be}$.}, we get the monomial decomposition:
$$
L(\bz + \bc) = \sum_{(i,\be,\bff) \in I} \bp_{i,\be}\bz^{\be} \otimes \bq_{i,\bff}\bz^{\bff}.
$$
Now, observe that the right hand side is a finite sum since $I$ is finite.  \\ 


\item {\bf Extracting subspaces spanned by coefficient vectors:}
We claim that we can take $U_1 = \spa (\bp_{i,\be}\ | \deg(\be) \leq \lfloor\frac{d}{2}\rfloor)$ and $U_2 = \spa (\bq_{i,\bff} \ | \ \deg(\bff) \leq d - (\lfloor\frac{d}{2}\rfloor +1))$ to satisfy the hypothesis of Lemma~\ref{basic}. Indeed, take $s \in S$. For some $\balpha \in \F^n$, we have $\phi(s) = L(\balpha+\bc) = \sum_{(i,\be,\bff) \in I} \bp_{i,\be}\balpha^{\be} \otimes \bq_{i,\bff}\balpha^{\bff}$. Take one of the terms $\bp_{i,\be}\balpha^{\be} \otimes \bq_{i,\bff}\balpha^{\bff}$. Since $\deg(\be) + \deg(\bff) \leq d$, we must have either $\deg(\be)\leq \lfloor \frac{d}{2} \rfloor$ or $\deg(\bff) \leq d - (\lfloor\frac{d}{2}\rfloor +1)$. If $\deg(\be) \leq \lfloor\frac{d}{2}\rfloor$, then 
$$
\bp_{i,\be}\balpha^{\be} \otimes \bq_{i,\bff}\balpha^{\bff} = \bp_{i,\be} \otimes (\bq_{i,\bff}\balpha^{\be+ \bff}) \in \bp_{i,\be} \otimes \F^l \subseteq U_1 \otimes \F^l.
$$
 Otherwise, $\deg(\bff) \leq  d - (\lfloor\frac{d}{2}\rfloor + 1)$, and $\bp_{i,\be} \balpha^{\be} \otimes \bq_{i,\bff}\balpha^{\bff} \in \F^k \otimes U_2$ . This means that $\phi(s) = L(\balpha+\bc) \in U_1 \otimes \F^l + \F^k \otimes U_2$ as desired. Let $J$ denote the set of monomials of degree $\leq \lfloor\frac{d}{2}\rfloor$ in $n$ variables. Then the defining spanning set of $U_1$ is indexed by $\{1,2,\dots,a\} \times J$. Hence $\dim(U_1) \leq aY_{n,d}$, since $|J| = Y_{n,d}$ by definition of $Y_{n,d}$. Similarly $\dim(U_2) \leq a Z_{n,d}$. Applying Lemma~\ref{basic} gives Theorem~\ref{barr-egow-war}.
\end{enumerate}
\end{proof}

The proof of Theorem~\ref{barr-egow-ten} is similar. The parametrization of rank $1$-tensors is of degree $d$, but one can additionally observe that the parametrization is `set-multilinear'. This forces additional constraints in the finite monomial decomposition, giving a sharper bound on potency. The full details can of course be found in~\cite{EGOW18}. We also need to discuss the notion of set-multilinearity for our purposes, but we defer that discussion until necessary.

\subsection{New ideas from algebraic geometry}
In adapting this proof strategy to prove barriers for generalized rank methods, the first issue occurs in Step $(2)$. The crucial point in Step $(2)$ is that rank of matrices can be described by polynomial equations, namely the vanishing of minors. This allows one to easily prove $\rk_{\F(\bz)}L(\bz) = \max\{\rk_\F L(\balpha) \ | \ \balpha \in \F^n\}$\footnote{All this requires is that $\F$ is infinite, or even sufficiently large.}, which is what allows for the symbolic decomposition over the function field. This argument fails if we consider a $T_k$-rank method or a $W_k$-rank method because tensor rank and Waring rank are not captured by polynomial conditions. 

Further, step $(3)$ runs into trouble as without some sort of symbolic decomposition, one cannot hope for a power series decomposition. Both these issues need to be addressed, and to do so, we will turn towards algebraic geometry.



We fix these issues in two steps (which when put together give Theorem~\ref{main-ag}). Roughly speaking, the first fixes step $(2)$ and the second fixes step $(3)$. We will describe these steps, and defer the proofs to Section~\ref{Sec:gen2symb}. The key idea in the first step is that one must pass from the function field to its algebraic closure. 
\begin{proposition} \label{prop:gensymb}
Let $\F$ be an algebraically closed field. Suppose $L: \F^n \rightarrow \F^m$ is a polynomial map, and $M: \F^r \rightarrow \F^m$ is another polynomial map such that $\im(L) \subseteq \im(M)$. Let $\bz = (z_1,\dots,z_n)$ be indeterminates and $\K = \F(z_1,\dots,z_n)$. Then there are algebraic functions $b_1(\bz),\dots,b_r(\bz) \in \overline{\K} = \overline{\F(z_1,\dots,z_n)}$ such that 
$$
L(\bz) = M(b_1(\bz),\dots,b_r(\bz)).
$$
\end{proposition}

The need to pass to the algebraic closure is already evident in the example at the end of this section. One ought to see the above as an algebraic result in a similar vein to the implicit function theorem. However, unlike the local nature of the implicit function theorem, this statement is more global\footnote{Branches of algebraic functions can be defined over a (large) Zariski open subset of the domain, but not necessarily the whole domain.}. Further, the implicit function theorem usually requires some non-degeneracy condition to be satisfied, and this is not the case for the above result (but we do have extra structure).

Our eventual goal is really to get power series rather than algebraic functions. When $\F = \C$, we can get power series by interpreting (an appropriate branch of) an algebraic function as an analytic function. The analogous statement holds for any algebraically closed field $\F$ of characteristic zero, but formulating and proving this requires some care. In particular, the notion of analytic does not exist, so we use the notion of \'etale morphisms as a suitable replacement. Before stating the second step, we will first recall the ring of power series.

Let $\bz = (z_1,\dots,z_n)$ be a vector of indeterminates. Recall from Section~\ref{sec:prelim} that $\F[|\bz - \bc|]$ denotes the ring of power series in the variables $\bz = (z_1,\dots,z_n)$ around $\bc \in \F^n$


\begin{proposition}\label{prop-ps}
Let $\F$ be an algebraically closed field of characteristic zero. Suppose we have a finite collection of elements 
$b_1(\bz),\dots,b_r(\bz) \in \overline{\K} = \overline{\F(z_1,\dots,z_n)}$. 
Then for some choice of $\bc \in \F^n$ (a generic choice will do), we have an $\F$-algebra homomorphism
$$
\F[z_1,\dots,z_n,b_1(\bz),\dots,b_r(\bz)] \longrightarrow \F[| \bz - \bc |],
$$
which extends the canonical inclusion 
$\F[z_1,\dots,z_n] \hookrightarrow \F[| \bz -\bc |]$.
\end{proposition}

Let us illustrate Theorem~\ref{main-ag} in a very simple case. Let $L:\F \rightarrow \F$ be given by $L(x) = x$. Let $M:\F \rightarrow \F$ be given by $M(y) = y^2$. Then, since $L$ and $M$ are surjective, it is clear that $\im(L) \subseteq \im(M)$, i.e., the hypothesis of Theorem~\ref{main-ag} is satisfied. If $z$ is an indeterminate, then $L(z) = z = M(\sqrt{z})$. Note that $\sqrt{z} \in \overline{\F(z)}$. In particular, this demonstrates the need to pass to the algebraic closure of the function field in Proposition~\ref{prop:gensymb}. Now, consider $\sqrt{z}$. This does not have a power series around $0 \in \F$. However, it will have a power series around some other point, say $1 \in \F$ (as claimed by Proposition~\ref{prop-ps}). To get this power series expansion, we expand $\sqrt{z} = (1 + (z-1))^{1/2}$ using the well known binomial theorem, to get
$$
\sqrt{z} = 1 + \frac{1}{2}(z-1) + \frac{-1}{8} (z-1)^3 + \dots.
$$
So, we have 
$$
L(z) = z = M\left(1 + \frac{1}{2}(z-1) + \frac{-1}{8} (z-1)^3 + \dots\right).
$$
Or equivalently, we get
$$
L(z + 1) = z+1 = M\left(1 + \frac{1}{2}z + \frac{-1}{8} z^3 + \dots\right)
$$
as claimed by Theorem~\ref{main-ag}.


\section{Numeric to symbolic transfer} \label{Sec:gen2symb}
This section will be devoted to developing the necessary tools from algebraic geometry, and using them to prove the numeric to symbolic transfer statement, i.e., Theorem~\ref{main-ag}. We will begin with some basic definitions. 

Let $\F$ be an algebraically closed field. For a finitely generated $\F$-algebra $A$, we denote by $\MSpec (A)$ the corresponding affine variety (over $\F$). As a set, $\MSpec(A)$ consists of all the maximal ideals of $A$. We further give it a topology called the Zariski topology by defining which subsets are closed. A subset of $\MSpec(A)$ is closed if it is of the form $\mathbb{V}(I) = \{m \in \MSpec(A)\ |\ I \subseteq m\}$ for some ideal $I$ of $A$.

Since $\F$ is algebraically closed, there is another description of $\MSpec(A)$ as $\F$-algebra homomorphisms from $A$ to $\F$. We denote by $\Hom(A,\F)$ the set of $\F$-algebra homomorphisms from $A$ to $\F$. Indeed, consider the map $\zeta: \MSpec(A) \rightarrow \Hom(A,\F)$ defined by the canonical quotient map $m \mapsto \{\zeta(m) : A \longrightarrow A/m = \F\}$. Note that since $A$ is a finitely generated $\F$-algebra, and $\F$ is algebraically closed, there is a canonical isomorphism $A/m = \F$. In the other direction, consider the map $\eta: \Hom(A,\F) \rightarrow \MSpec(A)$ defined by $\phi \mapsto \Ker(\phi)$. We leave it to the reader to check that the two maps are inverses to each other.

\begin{lemma} \label{max-hom}
The maps $\zeta$ and $\eta$ are inverses to each other. In particular, we have a canonical bijection between $\MSpec(A)$ and $\Hom(A,\F)$.
\end{lemma}

Suppose $A,B$ are finitely generated $\F$-algebras with an $\F$-algebra homomorphism $\iota: A \rightarrow B$. Then this gives a map $\iota^*: \Hom(B,\F) \rightarrow \Hom(A,\F)$ by $\phi \mapsto \phi \circ \iota$. Using the above lemma, we will also think of $\iota^*$ as a map from $\MSpec(B)$ to $\MSpec(A)$, and this is continuous with respect to the Zariski topology.




\subsection{Symbolic decomposition in terms of algebraic functions}
We will prove Proposition~\ref{prop:gensymb} in this subsection. The proof will be based on Hilbert's nullstellensatz.

\begin{proof} [Proof of Proposition~\ref{prop:gensymb}]
Let $L_i$ (resp. $M_i$) denote the coordinate functions of $L$, i.e., $L = (L_1,\dots,L_m)$ (resp. $M = (M_1,\dots,M_m))$. Let $\by = (y_1,\dots,y_r)$ be a vector of indeterminates. The hypothesis can be interpreted as follows -- for all $\balpha =  (\alpha_1,\dots,\alpha_n) \in \F^n$, the system of equations $\{L_i(\alpha_1,\dots,\alpha_n) = M_i(y_1,\dots,y_r)\}_{1\leq i \leq m}$ has a solution. 

Assume for the sake of contradiction that there are no $b_i(\bz) \in \overline{\K}$ such that $L(z_1,\dots,z_n) = M(b_1(\bz),\dots,b_r(\bz))$. This means that the system of equations 
$$
\{L_i(z_1,\dots,z_n) = M_i(y_1,\dots,y_r)\}_{1\leq i \leq m}
$$
has no solution. Just to be clear, we interpret these as $m$ equations in the indeterminates $y_1,\dots,y_r$ and coefficients in $\overline{\K}$. In other words, the zero locus of the collection of polynomials $\{M_i(\by) - L_i(\bz)\}_{1 \leq i \leq m} \subseteq \overline{\K}[y_1,\dots,y_r]$ is empty. By Hilbert's nullstellensatz, we get that
$$
\sum_i f_i(\by) \cdot (M_i(\by) - L_i(\bz)) = 1 \ (\text{in } \overline{\K}[y_1,\dots,y_r]),
$$
for some $f_i \in \overline{\K}[y_1,\dots,y_r]$. Write each $f_i = \sum_{\be \in \N^r} f_{i,\be}\by^{\be}$. Let $T := \F[z_1,\dots,z_n,f_{i,\be} \forall i,\be]$, which is a finitely generated ring.

Observe that the above equality can be intepreted in $T[y_1,\dots,y_r] \subseteq \overline{\K}[y_1,\dots,y_r]$.  Take any $\F$-algebra homomorphism $\phi:T \rightarrow \F$. That such a homomorphism exists is a consequence of Lemma~\ref{max-hom} and the fact that maximal ideals always exist. We can extend $\phi$ to a map $T[y_1,\dots,y_r] \rightarrow \F[y_1,\dots,y_r]$ which we will also call $\phi$ by abuse of notation. 

Let $\phi(z_i) = \beta_i \in \F$, and let $\bbeta = (\beta_1,\dots,\beta_n) \in \F^n$.  By applying $\phi$ to the above equality, we get
$$
\sum_i \phi(f_i) (M_i(\by) - L_i(\bbeta)) = 1 \ (\text{in } \F[y_1,\dots,y_r]) .
$$
which again by Hilbert's nullstellensatz means that the system of equations $\{L_i(\beta_1,\dots,\beta_n) = M_i(y_1,\dots,y_r)\}_{1\leq i \leq m}$ has no solution, which contradicts the hypothesis.
\end{proof}

\begin{remark}
An alternate proof using more modern algebro-geometric language is as follows. Consider $L$ as map of schemes rather than varieties, i.e., $L: \A^n \rightarrow \A^m$, where $\A^m$ defines the $m$-dimensional affine space (over $\F$). Similarly, consider $M$ also as a map of schemes. The hypothesis then tells us that $L(p) \in \im(M)$ for every closed point $p \in \A^n$. From this, one deduces that for the generic point $\eta \in \A^n$, $L(\eta) \in \im(M)$. This means that the fiber $M^{-1}(L(\eta))$ is non-empty. The symbolic vector $L(z_1,\dots,z_n)$ can be interpreted as a $\overline{\K}$-point lying over $\eta$. Since $M^{-1}(L(\eta))$ is non-empty, one can deduce that there is a closed $\overline{\K}$-point of $\A^r$ which is sent to $L(z_1,\dots,z_n)$ by $M$ (this uses that $\overline{\K}$ is algebraically closed). This just means that there is $(b_1(\bz),\dots,b_r(\bz)) \in (\overline{\K})^r$ such that 
$$
L(z_1,\dots,z_n) = M(b_1(\bz),\dots,b_r(\bz)).
$$
\end{remark}

\subsection{Power series representations of algebraic functions}
This subsection will be devoted to proving Proposition~\ref{prop-ps}. The intuition for the result is as follows. We want to give power series for each of the $b_i(\bz)$'s. Roughly speaking, power series around some point are analytic functions in some small (analytic) neighbourhood. So, we want to interpret all the $b_i(\bz)$ as analytic functions locally. The $b_i(\bz)$'s are in $\overline{\K}$ are not `functions' in $n$ variables -- at best they can be interpreted as `multi-valued' functions. Take for example the algebraic function $\sqrt{z}$ discussed at the end of Section~\ref{sec:new ideas}. At any non-zero point, there are two possible values for $\sqrt{z}$, and there is no canonical choice\footnote{In this particular case, there is a canonical choice over $\R$, but for us $\F$ is an algebraically closed field.}. On the other hand, there is a natural algebraic variety on which the $b_i(\bz)$'s are naturally functions. This variety is $\MSpec(R)$ where $R = \F[z_1,\dots,z_n,b_1(\bz),\dots,b_r(\bz)]$. 

Now, that the $b_i(\bz)$'s have been interpreted as functions on some algebraic variety, we observe that there is a morphism of varieties $f: \MSpec(R) \rightarrow \F^n$ given by the inclusion $\F[z_1,\dots,z_n] \hookrightarrow R$. Using the map $f$, we want to `push down' the functions $b_i(\bz)$ to functions on $\F^n$ locally. If we can do this, then we can interpret the $b_i(\bz)$ as functions in some small analytic neighbourhood of $\F^n$, which gives them power series. 

\begin{center}
\includegraphics[scale=0.5]{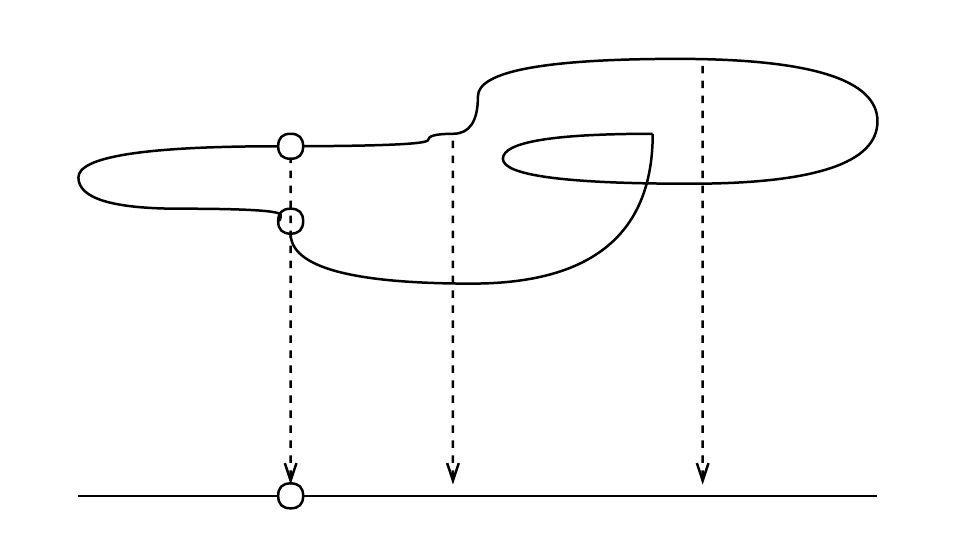}
\end{center}

The (toy) picture to have in the mind is the above one where the map $f$ is pictorially represented by projecting down (along the dotted arrows). The top curved picture represents $\MSpec(R)$ and the line at the bottom represents $\F^n$. Around most points in the domain, the map $f$ is an (analytic) isomorphism in some small local neighbourhood (in the analytic topology). So, using such a local isomorphism, we can interpret the $b_i(\bz)$'s as analytic functions in a small neighbourhood of $\F^n$, thereby giving them a power series.

Of course, as we are working with algebraic varieties, we do not have the analytic topology, but rather the very coarse Zariski topology. The correct notion to fix this issue is the notion of \'etale morphisms. We refer the reader to \cite{Stacks, Hartshorne, Ravi-Vakil} for standard results.

We will now proceed to give a rigorous proof of Proposition~\ref{prop-ps}. As usual, let $\bz = (z_1,\dots,z_n)$ denote indeterminates. In this section, $\F$ will be an algebraically closed field of characteristic zero. We note that the characteristic zero assumption seems to be crucial.

\begin{lemma}
Define the ring $R = \F[z_1,\dots,z_n,b_1(\bz),\dots,b_r(\bz)] \subseteq \overline{\K}$. Then $\dim(\MSpec(R)) = n$.
\end{lemma}

\begin{proof}
We have an inclusion $\F[z_1,\dots,z_n] \subseteq R \subseteq \overline{\K}$. Thus $R$ is an integral domain and a finitely generated $\F$-algebra. Thus, its Krull dimension is equal to its transcendence degree. But we know $\F(z_1,\dots,z_n) \subseteq {\rm Frac(R)} \subseteq \overline{\K}$, so the transcendence degree is $n$. Thus $\dim(\MSpec(R)) = n = \dim(\MSpec(\F[z_1,\dots,z_n])$. 
\end{proof}

Recall that an inclusion of rings $\F[z_1,\dots,z_n] \hookrightarrow R$ gives a dominant map\footnote{This just means that the image is dense.} $f: \MSpec(R) \rightarrow \F^n = \MSpec(\F[z_1,\dots,z_n])$. We say that the map $f$ is \'etale at a point $p \in \MSpec(R)$ if it is smooth at $p$ and relative dimension zero at $p$, i.e., the fiber $f^{-1}(f(p))$ is $0$-dimensional. We refer to the standard sources \cite{Stacks, Hartshorne, Ravi-Vakil} for the definition of smoothness of a morphism.

\begin{lemma}
There is a point $p \in \MSpec(R)$ at which $f$ is \'etale. 
\end{lemma}

\begin{proof}
Since we are in characteristic zero, we have generic smoothness on the source\footnote{This is a highly non-trivial result, and the analog of Sard's theorem.}, see for example \cite[Thm 25.3.1]{Ravi-Vakil}. This means that there is a dense open subset $U \subseteq \MSpec(R)$ such that $f$ is smooth on $U$. Further, since $f$ is dominant, there is an open subset $V \subseteq \MSpec(R)$ where $f$ has relative dimension zero (see~\cite[Proposition~11.4.1]{Ravi-Vakil}). Let $p \in U \cap V$. Then $f$ is smooth at $p$ with relative dimension zero at $p$, i.e., $f$ is \'etale at $p$.
\end{proof}


For a point $p \in \MSpec(R)$, we denote by $\OO_{R,p} := R_{p}$ the local ring at $p$. Roughly speaking the local ring at $p$ is the ring of functions which are polynomial in some small neighbourhood of $p$. We denote by $\widehat{\OO}_{R,p}$ the completion of $\OO_{R,p}$ with respect to the maximal ideal $p\OO_{R,p}$\footnote{For a ring $R$, its completion with respect to an ideal $I$ is the inverse limit $\widehat{T} := \lim\limits_{\leftarrow} R/I^n$.}. We have a canonical homomorphism from $R \rightarrow \OO_{R,p} \rightarrow \widehat{\OO}_{R,p}$.

\begin{lemma}
Let $p \in \MSpec(R)$ be a point at which $f$ is \'etale. Let $f(p) = \bc \in \F^n$. 
Then the inclusion $\F[z_1,\dots,z_n] \hookrightarrow R$ gives an isomorphism on the completions of local rings 
$$\F[| \bz-\bc |] \xrightarrow\sim \widehat{\OO}_{R,p}.$$
\end{lemma}

\begin{proof}
It is a well known result that if $f$ is \'etale  at $p$, then we have an isomorphism on the completions of the local rings (at $p$ and $f(p)$), see for example \cite[Theorem~1.6]{JTnotes} or \cite[Proposition~17.6.3]{EGA4}.  The completion of the local ring at $\bc \in \F^n$ is $\F[| \bz-\bc |]$. Thus we get an isomorphism 
$$
\F[| \bz - \bc|] \xrightarrow\sim   \widehat {\OO}_{R,p}.
$$
\end{proof}

\begin{proof} [Proof of Proposition~\ref{prop-ps}]
The above discussion can be summarized in the following commutative diagram.
\begin{center}
\begin{tikzcd}
\F[\bz] \arrow[hookrightarrow,r] \arrow[hookrightarrow,d]
& R \arrow[d] \\
\F[| \bz-\bc |] \arrow[r, "\sim" ]
& \widehat{\OO}_{R,p}
\end{tikzcd}
\end{center}

By inverting the lower horizontal isomorphism, we get the following commuting diagram, which is all we wanted.

\begin{center}
\begin{tikzcd}
\F[\bz] \arrow[hookrightarrow,r] \arrow[hookrightarrow,d]
& R \arrow[dl] \\
\F[| \bz-\bc |]
\end{tikzcd}
\end{center}
\end{proof}

\begin{proof} [Proof of Theorem~\ref{main-ag}]
From Proposition~\ref{prop:gensymb}, we get $L(z_1,\dots,z_n) = M(b_1(\bz),\dots,b_r(\bz)).$ Then, we can apply the homomorphism $\F[z_1,\dots,z_n,b_1(\bz),\dots,b_r(\bz)] \longrightarrow \F[| \bz - \bc |]$ given by Proposition~\ref{prop-ps}. This replaces the $b_i(\bz)$ by power series around $\bc$ for each $i$. Then, applying the shift $\bz \mapsto \bz + \bc$, we get the required conclusion. 
\end{proof}

\begin{proof} [Proof of Corollary~\ref{cor:ag}]
This is a special case of Theorem~\ref{main-ag}. To see this, we only need to find the right $M$. Take $M : ((\F^m)^k)^a \rightarrow \Ten(m,k)$ given by 
$$
M((\bv^{(1)}_1,\dots,\bv^{(k)}_1),\dots,(\bv^{(1)}_a,\dots,\bv^{(k) }_a)) = \sum_{i=1}^a \bv^{(1)}_i \otimes \bv^{(2)}_i \otimes \dots \otimes \bv^{(k)}_i,
$$
where all $\bv_i^{(j)} \in \F^m$. Observe that $M$ parametrizes the set of all tensors of rank $\leq a$. Hence, we have $\im(L) \subseteq \im(M)$. Thus, we can apply Theorem~\ref{main-ag} and deduce the required result.
\end{proof}

\section{Tensor rank lower bound methods}\label{sec:tensor-rank}

In this section, we will prove upper bounds on the potency of $T_k$-rank methods. In the previous sections, we have discussed the necessary algebraic geometry that allows us to parallel the arguments in \cite{EGOW18}.

Let us first collect some general statements. Let $S \subseteq V$ be a spanning set (of simples). Suppose we have a  linear map $\phi: V \rightarrow \Ten(m,k)$, i.e., a $T_k$-rank method.

\begin{definition}
Suppose $U \subseteq \Ten(m,k)$ is a linear subspace. Then we define
$$
r(U) := \max \{ \trk(T) \ | T \in U\}.
$$
\end{definition}

The following lemma is straightforward, and is in the same spirit as Lemma~\ref{basic}.

\begin{lemma}\label{basic2}
Suppose $\phi(S) \subseteq \sum_{\I} U_{\I}$ for a finite collection of linear subspaces $U_{\I}$, then $\phi(V) \subseteq \sum_{\I} U_{\I}$. Thus, we have
$$
\mu_\phi(V) \leq \sum_{\I} r(U_{\I}).
$$
\end{lemma}

Now, let us consider some special subspaces. 
\begin{definition} \label{basic-sub}
A subspace $U \subseteq \Ten(m,k)$ is called {\em basic} if there exist subspaces $U_i \subseteq \F^m$ for $1 \leq i \leq k$ such that  
$$
U = U_1 \otimes U_2 \otimes \dots \otimes U_k \subseteq \Ten(m,k).
$$
\end{definition}

\begin{lemma}\label{basic2r}
Suppose $U = U_1 \otimes U_2 \otimes \dots \otimes U_k \subseteq \Ten(m,k)$ is a basic subspace. Then for any $p$ such that $1 \leq p \leq k$, we have
$$
r(U) \leq \prod_{i \neq p} \dim(U_i).
$$

\end{lemma}

The idea for barriers is as follows. We will be given a rank method $\phi$ with $\mu_\phi(S) = a$. We will find a collection of basic subspaces $U_{\I}$ that satisfy the hypothesis of Lemma~\ref{basic2}. This will give an upper bound on $\mu_\phi(V)$ in terms of $r(U_{\I})$'s, which in turn can be computed by Lemma~\ref{basic2r}. Since $\Pot(\phi) = \frac{\mu_\phi(V)}{a}$, the upper bound on $\mu_\phi(V)$ gives an upper bound on potency as well. So, all that is left to do is to exhibit the required collections of basic subspaces in the cases that we need to prove.

\subsection{Set multi-grading}
We will collect some notation and basic facts regarding set multi-degree that we will use subsequently to prove our main barrier results. Let $\bz = (\bz_1, \ldots, \bz_d)$ be a set of variables, where each set
$\bz_i = (z_{i1}, \ldots, z_{in})$ corresponds to the $i^{th}$ set of 
variables in $\bz$. We can identify the set of monomials in $\bz$ with their exponent vectors
$\be \in (\N^n)^d$, that is $\bz^\be \leftrightarrow \be$, and we will do so whenever convenient.

We will define a set multi-grading on the polynomial ring $\F[\bz]$, which is an $\N^d$-valued grading. This grading will be a refinement of the grading by (total) degree. The indeterminate $z_{ij}$ will be set multi-homogeneous, and its set multi-degree will be $\smdeg(z_{ij}) = \bdelta_i = (0,\dots,0,\underbrace{1}_{i},0,\dots,0)$ (i.e., a 1 in the $i^{th}$ position). Hence, any monomial $\bz^{\be} = \prod_{i,j} z_{ij}^{e_{ij}}$ will have
$$
\smdeg(\bz^\be) = \smdeg(\be) = \sum_{i,j} e_{ij} \bdelta_i = (\sum_{j} e_{1j}, \sum_j e_{2j}, \dots, \sum_j e_{dj}).
$$ 

We now give the formal definition.
\begin{definition}[Set multi-grading]
We define an $\N^d$-grading on $\F[\bz]$ called the set multi-grading by writing $\F[\bz] = \oplus_{\bff \in \N^d} \F[\bz]_{\bff}$, where $\F[\bz]_{\bff}$ is spanned by monomials $\bz^\be$ such that $(\sum_{j} e_{1j}, \sum_j e_{2j}, \dots, \sum_j e_{dj})= \bff$. For a set multi-homogeneous polynomial $q(\bz) \in \F[\bz]_{\bff}$, we write $\smdeg(q(\bz)) = \bff$, and call this the set multi-degree of $q(\bz)$.
\end{definition}

Define a partial order on $\N^d$ given by $(a_1,\dots,a_d) \preceq (b_1,\dots,b_d)$ if 
$a_i \leq b_i$ for $1 \leq i \leq d$ (in which case we say $(a_1,\dots,a_d)$ is smaller than 
$(b_1,\dots,b_d)$).

{\bf Zero-one vectors and subsets.} 
A zero-one vector in $\N^d$ is a vector whose coordinates are either zero or one. There is a natural correspondence between zero-one vectors in $\N^d$ with subsets of $[d]$. We will now make this precise. First, given a zero-one vector in $\N^d$, we define its support subset.

\begin{definition} [Support subset] \label{supp.subset}
Given a zero-one vector $\bff = (f_1,\dots,f_d) \in \N^d$, we define its support subset 
$\supp(\bff) := \{i \in [d]\ |\ f_i \neq 0\} =  \{i \in [d]\ |\ f_i = 1\} \subseteq [d]$
\end{definition}

In the other direction, we define indicator vector to a subset. 

\begin{definition} [Indicator vector] \label{ind.vector}
Given a subset $J \subseteq [d]$, 
we define its indicator vector $\bdelta_J $ to be the zero-one vector which has a $1$ in the $i^{th}$ position if and 
only if $i \in J$. 
\end{definition}

\begin{align*}
\left\{\text{Zero-one vectors } \in \N^d \right\} & \longleftrightarrow  \left\{\text{subsets of } [d] \right\} \\
\bff & \longrightarrow  \supp(\bff) \\
\bdelta_J &\longleftarrow J
\end{align*}

The above correspondence takes a zero-one vector to its support subset, and in the other direction takes a subset to its indicator vector. Moreover, using the partial order $\preceq$ defined above, we also get:
$$
\left\{\bff \in \N^d\ | \ \bff \preceq (1,1,\dots,1) \right\} = \left\{\text{Zero-one vectors } \in \N^d \right\} \longleftrightarrow \left\{\text{subsets of } [d] \right\}.
$$

{\bf Set Partitions.}
A set partition of $[d] := \{1,2,\dots,d\}$ into $k$ sets is a tuple of subsets $\I = (I_1,\dots,I_k)$, 
where $I_1 \sqcup I_2 \sqcup \dots \sqcup I_k = [d]$. Let $\SP(d,k)$ denote all set partitions of 
$[d]$ into $k$ sets. Note that $|\SP(d,k)| = k^d$.

Using the dictionary between subsets and zero-one vectors, we get the following correspondence.

$$
\left\{ \bff^{(1)},\dots,\bff^{(k)} \in \N^d \ | \ \sum_j \bff^{(j)} = (1,1,\dots,1) \right\} \longleftrightarrow \SP(d,k)
$$

\subsection{Barriers for tensor rank lower bounds}
This whole section parallels the proof sketch of Theorem~\ref{barr-egow-war} given earlier in Section~\ref{sec:new ideas}. There are a few differences. The first is that we need to consider the more refined notion of set multi-degree as opposed to the usual notion of (total) degree. The second is that we add a homogenization step, which exists in the original proof in \cite{EGOW18} (but was not strictly necessary). In this case, however, the homogenization step significantly simplifies the proof, so we include it. Finally, we replace the second and third steps of the proof sketch of Theorem~\ref{barr-egow-war} in one go, by using Corollary~\ref{cor:ag}. 

We point out explicitly the analogous steps: Equation~(\ref{L-power series}) is the power series decomposition, and Equation~(\ref{L-truncated}) is the finite monomial decomposition. From this monomial decomposition, we extract subspaces spanned by coefficient vectors (Lemma~\ref{bubbly}). Finally, we compute the upper bound on potency we get from applying Lemma~\ref{basic2r}.

Let $\psi: (\F^n)^{\times d} \rightarrow \Ten(n,d)$ be the parametrization of $S$ (the set of rank $1$ tensors) given by $(\bv_1,\dots,\bv_d) \mapsto \bv_1 \otimes \bv_2 \otimes \dots \otimes \bv_d$, where $\bv_i \in \F^n$. Let $\bz = (\bz_1, \ldots, \bz_d)$ be a set of variables, where each set
$\bz_i = (z_{i1}, \ldots, z_{in})$ corresponds to the $i^{th}$ set of 
variables in $\bz$. Let $\phi: \Ten(n,d) \rightarrow \Ten(m,k)$ be a linear map, i.e., a $T_k$-rank method. Let $L = \phi \circ \psi: (\F^n)^{\times d} \rightarrow \Ten(m,k)$.

Observe that $\psi(\bz)$ is a tensor whose entries are set multi-homogeneous polynomials (in $\bz$) of set multi-degree $(1,1,\dots,1) \in \N^d$. Since $\phi$ is linear, $L(\bz)$ is also a tensor whose entries are set multi-homogeneous polynomials (in $\bz$) of set multi-degree $(1,1,\dots,1)$.

Let $a = \mu_\phi(S) = \max\{\rk(\phi(s)) \ |\ s \in S\}$. Observe that $\im(L) = \phi(S)$, so $\trk(L(\bbeta)) \leq a$ for all $\bbeta \in (\F^n)^{\times d}$. So, we can apply Corollary~\ref{cor:ag} to get a power series decomposition:

\begin{equation} \label{L-power series}
L(\bz+\bc) = \sum_{i=1}^a \bp_i^{(1)} (\bz) \otimes \bp_i^{(2)}(\bz) \otimes \dots \otimes \bp_i^{(k)}(\bz),
\end{equation}
for some $\bc \in (\F^n)^{\times d}$, where $\bp_i^{(j)}(\bz)$ are power series (around $\bo$).  Write each 
$$
\bp_i^{(j)}(\bz) = \sum_{\be \in (\N^n)^{d}} \bp_{i,\be}^{(j)} \bz^\be.
$$

Observe that we have:

$$
L(\bz) = L(\bz+\bc)_{(1,1,\dots,1)}= \sum_{i=1}^a (\bp_i^{(1)} (\bz) \otimes \bp_i^{(2)}(\bz) \otimes \dots \otimes \bp_i^{(k)}(\bz))_{(1,1,\dots,1)}.
$$

So, we can consider the relevant terms to get a monomial decomposition.

\begin{equation} \label{L-truncated}
L(\bz) = \sum_{i=1}^a \left(\sum_{\begin{array}{c} \be^{(1)},\dots,\be^{(k)} \\ \sum \smdeg( \be^{(j)}) = (1,1,\dots,1) \end{array}} \bp^{(1)}_{i,{\be^{(1)}}} \otimes \bp^{(2)}_{i,{\be^{(2)}}} \otimes \dots \otimes \bp^{(k)}_{i,{\be^{(k)}}} (\bz^{\be^{(1)} + \dots +\be^{(k)}})  \right)
\end{equation}


We will show that $\im(L) = \phi(S)$ is contained in a union of certain basic subspaces. To describe these basic subspaces we need to introduce a bit of notation.



Recall that $\SP(d,k)$ denotes the set of all partitions of $[d]$ into 
$k$ sets. Also recall that for any subset $I \subseteq [d]$, we define ${\bdelta}_I \in \N^d$ to be its indicator vector.
For $\I = (I_1,\dots,I_k) \in \SP(d,k)$ and $1 \leq i \leq a$, we define
$$
\mathcal{C}^i_{\I} := \mathcal{C}^i_{I_1} \otimes \mathcal{C}^i_{I_2} \otimes \dots \otimes \mathcal{C}^i_{I_k},
$$
where for all $j$, we have
$$
\mathcal{C}^i_{I_j} := \spa( \bp^{(j)}_{i,\be}  \mid \smdeg(\be) = \bdelta_{I_j} ) \subseteq \F^m.
$$

The following lemma is the crucial observation that $\phi(S)$ is contained in the sum of all the basic subspaces $\mathcal{C}^i_{\I}$ for all $\I \in \SP(d,k)$. 

\begin{lemma} \label{bubbly}
We have $\phi(S) = \im(L) \subseteq \sum_{i=1}^a \sum_{\I \in \SP(d,k)} \mathcal{C}^i_{\I}.$
\end{lemma}

\begin{proof}
For any $\balpha \in \F^{nd}$, we need to show that $L(\balpha) \in \sum_{i=1}^a \sum_{\I \in \SP(d,k)} \mathcal{C}^i_{\I}.$ Fix $\balpha \in \F^{nd}$. Plug in $\bz = \balpha$ into equation~(\ref{L-truncated}). This gives $L(\balpha)$ as a sum of many terms. It suffices to show that each term is in $\mathcal{C}^i_{\I}$ for some $\I \in \SP(d,k)$ and $1 \leq i \leq a$. To this end, take some term in the sum
$$
t = \bp^{(1)}_{i,{\be^{(1)}}} \otimes \bp^{(2)}_{i,\be^{(2)}} \otimes \dots \otimes \bp^{(k)}_{i,\be^{(k)}} (\balpha^{\be^{(1)} + \dots + \be^{(k)}}),
$$
with  $\sum_j \smdeg(\be^{(j)}) = (1,1,\dots,1)$.
Thus, all we need to do is to produce $\I$ such that $t \in \mathcal{C}_{\I}^i$. First, observe that $\sum_j \smdeg (\be^{(j)}) =  (1,1,\dots,1)$ means that $\smdeg ( \be^{(j)}) \in \N^d$ is a zero-one vector for all $j$. Recall that zero-one vectors correspond to subsets, and that the subset corresponding to a zero-one vector is called the support subset (see Definition~\ref{supp.subset}). Let $I_j  = \supp(\smdeg(\be^{(j)}) \subseteq [d]$ be the support subset of $\smdeg(\be^{(j)})$. It follows from $\sum_j \smdeg (\be^{(j)}) =  (1,1,\dots,1)$ that $I_1 \sqcup I_2 \sqcup \dots \sqcup I_k = [d]$. Thus, $\I = (I_1,\dots,I_k) \in \SP(d,k)$.

For all $j$, we have that $\bdelta_{I_j} = \smdeg(\be^{(j)})$ by definition of $I_j$. Hence, $\bp^{(j)}_{i,\be^{(j)}} \in \mathcal{C}^i_{I_j}$. Thus, the term $t \in \mathcal{C}^i_{\I}$ as required. Note that $\balpha^{\be^{(1)} + \dots + \be^{(k)}}$ is just a constant.
\end{proof}

Combining Lemma~\ref{basic2} with the above lemma, we get the following corollary.

\begin{corollary} \label{alia}
We have
$$
\Pot(\phi) = \frac{\mu_\phi(V)}{a} \leq \frac{\displaystyle\sum_{i=1}^a \left(\sum_{\I \in \SP(d,k)} r(\mathcal{C}^i_{\I})\right)}{a} . 
$$
\end{corollary}

Now, we have just one computation left.

\begin{lemma}\label{lem:tensor-count}
We have
$$
\sum_{\I \in \SP(d,k)} r(\mathcal{C}^i_{\I}) \leq k^d n^{\lfloor \frac{(k-1)d}{k} \rfloor}.
$$
\end{lemma}

\begin{proof}
We  have $|\SP(d,k)| = k^d$. So, it suffices to show that for each $r(\mathcal{C}^i_{\I}) \leq n^{\lfloor \frac{(k-1)d}{k} \rfloor}$ for all $\I \in \SP(d,k)$. We do this as follows. First, note that $\dim(\mathcal{C}^i_{I_j}) \leq n^{|I_j|} = $ number of monomials whose $\smdeg$ is $\bdelta_{I_j}$. 

Let $r$ be such that $|I_r| \geq |I_j|$ for all $j$. Then $\sum_{j \neq r} |I_j| \leq \lfloor \frac{(k-1)d}{k} \rfloor$. Thus, we have
$$
r(\mathcal{C}^i_{\I}) \leq \prod_{j \neq r} \dim(\mathcal{C}^i_{I_j}) \leq \prod_{j \neq r} n^{|I_j|} = n^{\sum_{j \neq r} |I_j|} \leq n^{\lfloor \frac{(k-1)d}{k} \rfloor}
$$
as required.
\end{proof}

\begin{proof} [Proof of Theorem~\ref{TtoT}]
This follows from combining Lemma~\ref{lem:tensor-count} with Corollary~\ref{alia}.
\end{proof}

\subsection{Barriers for Waring rank lower bounds}
This subsection follows a completely identical strategy to the previous one. The only difference is that we do not use the notion of set multi-grading, but the usual grading on polynomials given by total degree.

For this section, let $\psi: \F^n = P(n,1) \rightarrow P(n,d)$ denote the map $\ell \mapsto \ell^d$ for $\ell \in P(n,1)$. Then $\psi$ parametrizes the simples in this case, i.e, $\im(\psi) = S = \{\ell^d\ |\ \ell \in P(n,1)\}$. Let $\phi: P(n,d) \rightarrow \Ten(m,k)$ be a linear map, i.e., a $T_k$-rank method. Let $a = \max\{\trk(\phi(s)) \ |\ s \in S\} = \mu_\phi(S)$. Let $L = \phi \circ \psi: \F^n \rightarrow P(n,d)$. 

Note that $\psi$ is a homogeneous polynomial map of degree $d$, and $\phi$ is linear. So, $L$ is a homogeneous polynomial map of degree $d$. By Corollary~\ref{cor:ag}, we have the power series decomposition
$$
L(\bz+\bc) = \sum_{i=1}^a \bp_i^{(1)} (z) \otimes \bp_i^{(2)}(z) \otimes \dots \otimes \bp_i^{(k)}(z)
$$
for some $\bc \in \F^n$, and $\bp_i^{(j)}$ are power series around $\bo$. Write each $\bp_i^{(j)}(\bz) = \sum_{\be \in \N^n} \bp_{i,\be}^{(j)} \bz^\be.$ We have
$$
L(\bz) = L(\bz+\bc)_{d}= \sum_{i=1}^a (\bp_i^{(1)} (\bz) \otimes \bp_i^{(2)}(\bz) \otimes \dots \otimes \bp_i^{(k)}(\bz))_{d}.
$$

Hence, we get the finite monomial decomposition:

\begin{equation} \label{L-truncated2}
L(\bz) = \sum_{i=1}^a \left(\sum_{\begin{array}{c} \be^{(1)},\dots,\be^{(k)} \\ \sum \deg \be^{(i)} = d \end{array}} \bp^{(1)}_{i,\be^{(1)}} \otimes \bp^{(2)}_{i, \be^{(2)}} \otimes \dots \otimes \bp^{(k)}_{i,\be^{(k)}} (\bz^{\be^{(1)} + \dots + \be^{(k)}})  \right)
\end{equation}

Let $\Pp(d,k) = \{\mu = (\mu_1,\dots,\mu_k) \ |\ \sum_i \mu_i = d\}$ denote the set of ordered $k$-partitions of $d$. For $\mu \in \Pp(d,k)$, let us define
$$
\mathcal{C}^i_{\mu} = \mathcal{C}^i_{\mu_1} \otimes \mathcal{C}^i_{\mu_2} \otimes \dots \otimes \mathcal{C}^i_{\mu_k},
$$
where
$$
\mathcal{C}^i_{\mu_j} = \spa (\bp_{i,\be} \ |\ \deg \be = \mu_j\}.
$$

\begin{lemma}
We have 
$$
\phi(S) = \im(L) \subseteq \sum_{i=1}^a \sum_{\mu \in \Pp(d,k)} \mathcal{C}^i_{\mu}.
$$
\end{lemma}

\begin{proof}
This is similar to Lemma~\ref{bubbly}, so we omit the details.
\end{proof}

Combining with Lemma~\ref{basic2}, we have

\begin{corollary}
We have
$$
\Pot(\phi) \leq \frac{\mu_\phi(V)}{a} \leq \frac{\sum_{i=1}^a \sum_{\mu \in \Pp(d,k)} r(\mathcal{C}^i_{\mu})}{a} . 
$$
\end{corollary}

\begin{lemma}
For all $\mu \in \Pp(d,k)$, we have 
$$
r(\mathcal{C}^i_\mu) \leq \prod\limits_{j \neq l} {n + \mu_j - 1 \choose \mu_j}.
$$ where $l$ is such that $\mu_l \geq \mu_j$ for all $j$.
\end{lemma}

\begin{proof}
This just follows from Lemma~\ref{basic2r} and the fact that $\dim \mathcal{C}^i_{\mu_j} \leq {n + \mu_j - 1 \choose \mu_j}$, which is the size of its defining spanning set.
\end{proof}

Let us define 
$$
\Upsilon_\mu := \prod\limits_{j \neq l} {n + \mu_j - 1 \choose \mu_j},
$$
where $l$ is such that $\mu_l \geq \mu_j$ for all $j$. Since ${n + \mu_j - 1 \choose \mu_j} = O(n^{\mu_j})$, and $\sum_{j \neq l} \mu_j \leq  \lfloor \frac{(k-1)d}{k} \rfloor$, we have that 
$$
\Upsilon_\mu \leq O(n^{ \lfloor \frac{(k-1)d}{k} \rfloor }).
$$

\begin{corollary} \label{actual-war-bound}
We have $\Pot(\phi) \leq \displaystyle\frac{\sum_{i=1}^a \sum_{\mu \in \Pp(d,k)} r(\mathcal{C}^i_{\mu})}{a} \leq \sum_{\mu \in \Pp(d,k)} \Upsilon_\mu$.
\end{corollary}

\begin{proof} [Proof of Theorem~\ref{PtoT}]
This follows from the previous corollary since $\Upsilon_\mu \leq O(n^{ \lfloor \frac{(k-1)d}{k} \rfloor })$ as we saw above, and $|\Pp(d,k)|$ is just some constant that depends only on $d$ and $k$.
\end{proof}

\section{Waring rank lower bound methods} \label{sec:war-rank}
We will derive the upper bounds on the potency for $W_k$-rank methods from the upper bounds on the potency for $T_k$-rank methods. The upper bounds will be weaker, but the loss is a constant that depends only on $k$.

Let $S \subseteq V$ be a spanning subset (simples). Let $\phi: V \rightarrow P(m,k)$ be a linear map, i.e., a $W_k$-rank method. Let $\tilde{\phi}: V \rightarrow \Ten(m,k)$ be the composite map $\iota \circ \phi$, where $\iota: P(m,k) \hookrightarrow \Ten(m,k)$ is the natural inclusion of polynomials of degree $k$ as symmetric $k$-tensors. Let $S_k$ denote the symmetric group on $k$ letters. The group $S_k$ acts on $\Ten(m,k)$ by permuting the tensor factors. A tensor is called symmetric if it is invariant under this action. 

Let us describe the map $\iota$. First note that $P(m,1)$ is a vector space of dimension $n$, so we have an isomorphism $P(m,1) \rightarrow \K^m$ given by $\ell = \ell_1x_1 + \ell_2x_2 + \dots + \ell_mx_m \mapsto \bell = (\ell_1,\dots,\ell_m)$. In the following, we will use the identification $\ell \leftrightarrow \bell$ freely to represent the isomorphism. Using this identification, we can describe $\iota$ by describing it on monomials. 


$$
\iota(\ell^{(1)} \ell^{(2)}\cdots \ell^{(k)}) = \frac{1}{k!} \sum_{\sigma \in S_k} \bell^{(\sigma(1))} \otimes \bell^{(\sigma(2))} \otimes \dots \otimes \bell^{(\sigma(k))}.
$$

Note in particular that this means

$$
\iota(\ell^k) = \bell \otimes \bell \otimes \dots \otimes \bell.
$$

\begin{lemma}
For any $f \in P(m,k)$, we have 
$$
\trk(\iota(f)) \leq \wrk(f) \leq 2^{k-1} \trk(\iota(f)).
$$
\end{lemma}

\begin{proof}
Under the map $\iota$, a power of a linear form, i.e., $\ell^k$ is sent to a rank $1$ tensor. This means that a decomposition of $f$ as a sum of powers of linear forms is sent to a decomposition of $\iota(f)$ as a sum of rank $1$ tensors. This gives $\trk(\iota(f)) \leq \wrk(f)$. 

On the other hand, suppose $\iota(f) = \sum_{i=1}^r \bell^{(i1)} \otimes \dots \otimes \bell^{(ik)}$. Then since $\iota(f)$ is symmetric, we can write $\iota(f) = \sum_i \bell^{(i\sigma(1))} \otimes \dots \otimes \bell^{(i\sigma(k))}$ for any permutation $\sigma \in  S_k$. In particular, we have
$$
\iota(f) =  \frac{1}{k!} \sum_{\sigma \in S_k} \sum_{i=1}^r \bell^{(i\sigma(1))} \otimes \dots \otimes \bell^{(i\sigma(k))}.
$$
But this means that
$$
\iota(f) = \sum_{i=1}^r \iota(\ell^{(i1)}\ell^{(i2)} \dots \ell^{(ik)})
$$
Since $\iota$ is an injective (and linear) map, we deduce that
$$
f = \sum_{i=1}^r \ell^{(i1)} \ell^{(i2)} \dots \ell^{(ik)}
$$
Now, each term $\ell^{(i1)}\dots \ell^{(ik)}$ can be written as a sum of $2^{k-1}$ linear forms by Glynn's formula (\cite{Glynn}) that we recalled in Example~\ref{glynn-formula}. This gives $f$ as a sum of $2^{k-1}r$ powers of linear forms. In other words, we have $\wrk(f) \leq 2^{k-1} \trk(\iota(f))$.

\end{proof}

\begin{corollary} 
We have $\mu_{\tilde{\phi}}(S) \leq \mu_{\phi}(S)$.
\end{corollary}

\begin{proof}
From the above lemma, we know $\trk(\iota(f)) \leq \wrk(f)$ for all $f \in P(m,k)$. In particular, using this for every $f = \phi(s)$ for $s \in  S$, we see that $\trk(\tilde{\phi}(s)) = \trk(\iota(\phi(s)) \leq \wrk(\phi(s)).$
\end{proof}

A similar argument shows the following.

\begin{corollary}
We have $\mu_\phi(V) \leq 2^{k-1} \mu_{\tilde{\phi}}(V)$.
\end{corollary}

Combining the previous two corollaries, we get:

\begin{corollary}
We have $\Pot(\phi) \leq 2^{k-1} \Pot(\tilde{\phi}).$
\end{corollary}

\begin{proof} [Proofs of Theorems~\ref{TtoP} and ~\ref{PtoP}]
These follows from applying the above corollary to Theorems~\ref{TtoT} and ~\ref{PtoT}.
\end{proof}

\begin{remark} 
A famous conjecture of Comon was that $\wrk(f) = \trk(\iota(f))$ for any $f \in P(m,k)$, see \cite{CGLM08}. This was proved to be true in many special cases. Recently, a rather complicated counterexample has appeared in \cite{Shitov}. While this means that Comon's conjecture is false, the evidence would suggest that the inequality $\wrk(f) \leq 2^{k-1} \trk(\iota(f))$ is far from being sharp. It is an interesting question to find more optimal replacements for the factor of $2^{k-1}$.
\end{remark}

\section{Barriers for border rank methods}\label{sec:border-rank}
In this section we prove analogous theorems to Theorem~\ref{TtoT} and Theorem~\ref{PtoT}, but now
instead of proving upper bounds on the potency of tensor rank methods, we prove upper bounds on the potency of border rank methods for tensors. Roughly speaking, a border rank method will use a linear map to lift border rank lower bounds for low degree tensors to border rank lower bounds for higher degree tensors. While this is a natural generalization, we require additional tools from algebraic geometry to establish barriers for these border rank methods.

We will briefly recall some notions regarding border rank, and then prove the analogous barriers. The high level strategy remains the same, with more details to be worked out. The key new ideas are the notion of degenerations and the ability to switch between topological border rank and algebraic border rank. In this entire section, we will assume $\F$ is an algebraically closed field of characteristic zero.

We will need to work over the polynomial ring $\F[\veps]$ of polynomials over the variable $\veps$. We define $\Ten_{\F[\veps]}(m,k) := (\F[\veps]^m)^{\otimes k}$. We will say $T \in \Ten_{\F[\veps]}(m,k)$ is simple if $T = \bp_1(\veps) \otimes \bp_2(\veps) \otimes \dots \otimes \bp_k(\veps)$ for some $\bp_i(\veps) \in \F[\veps]^m$. For any tensor $T \in \Ten_{\F[\veps]}(m,k)$, we will write 
$$
\trk_{\F[\veps]}(T) = \min\{ r\ | \ T = T_1 + \dots + T_r, \ T_i \text{ simple}\}.
$$
We will write $\Ten_{\F[\veps]}(m,k)_{\leq r} :=\{T \ | \  \trk_{\F[\veps]}(T) \leq r\}$.

\begin{definition}[Topological Border Rank]\label{def:top-brk}
	 A tensor $T \in \Ten_\F(m,k)$ has border rank $\leq r$, denoted by $\brk(T) \leq r$ if
	 $T \in \overline{\Ten_\F(m,k)_{\leq r}}$, where the closure is the Zariski closure of the set
	 of rank $r$ tensors.
\end{definition}

The above definition defines border rank implicitly as $\brk(T) = \min\{r\ |\ T \in \overline{\Ten_{\F}(m,k)_{\leq r}} \}$. But we prefer the above definition for later use.

\begin{definition}[Degeneration~\cite{BCS13}]	
	Given a tensor $T \in \Ten_\F(m,k)$ we say that $T$ is a 
	{\em degeneration of order} $q$ of a rank $r$ tensor, denoted by 
	$T \unlhd_q \langle r \rangle$, if 
	there exist tensors $T_1 \in \Ten_{\F[\veps]}(m,k)_{\leq r}$ and 
	$T_2 \in \Ten_{\F[\veps]}(m,k)$ such that 
	$$ \veps^{q-1} \cdot T = T_1 + \veps^q \cdot T_2. $$
\end{definition}

\begin{definition}[Algebraic Border Rank~\cite{BCS13}]\label{def:alg-rk}
	We say that a tensor $T \in \Ten_\F(m,k)$ has algebraic border rank $\leq r$ if there exists a number 
	$q \in \Z_{\geq 1}$ such that $T \tleq_q \langle r \rangle $. 
\end{definition}

Again, this defines algebraic border rank implicitly.

\begin{theorem}[Theorem 20.24 in~\cite{BCS13} due to Strassen]\label{thm:top-eq-alg-rk}
	Definitions~\ref{def:top-brk} and~\ref{def:alg-rk} are equivalent. That is, given a tensor 
	$T \in \Ten_\F(m,k)$
	$$ \exists q \in \N \st T \tleq_q \langle r \rangle \iff T \in \overline{\Ten_\F(m,k)_{\leq r}}. $$
\end{theorem}

Note that $\F$ being algebraically closed is crucial for the above theorem to hold. By default, border rank will mean topological border rank.

 A subset $U \subseteq \F[\veps]^{\ell}$ is called a $\F[\veps]$-submodule if it is closed under addition and multiplication by elements of $\F[\veps]$. A subset $\{\bp_1(\veps),\bp_2(\veps),\dots \bp_k(\veps)\} \subseteq \F[\veps]^{\ell}$ is a generating set for the submodule $U$ if $\forall T \in U$, one can write $T = \sum_i \bp_i(\veps) c_i(\veps)$ for some $c_i(\veps) \in \F[\veps]$. 
\begin{definition}
For an $\F[\veps]$-submodule $U \subseteq \F[\veps]^m$, we define its rank $\rk_{\F[\veps]}(U)$ to be the size of its smallest generating set.
\end{definition}

Note that $\Ten_{\F[\veps]}(m,k) \cong \F[\veps]^{m^k}$, so it makes sense to talk about submodules of $\Ten_{\F[\veps]}(m,k)$ We will consider degenerations to modules, as the following definition alludes to.

\begin{definition}[Degeneration to a Submodule]
	We say that a tensor $T \in \Ten_{\F[\veps]}(m,k)$ degenerates to a module 
	$U \subseteq \Ten_{\F[\veps]}(m,k)$ with order $q$, written $T \tleq_q U$ if there exist tensors
	$T_1 \in U$ and $T_2 \in \Ten_{\F[\veps]}(m,k)$ such that 
	$$ \veps^{q-1} \cdot T = T_1 + \veps^q \cdot T_2. $$
	More generally, we say that a subset $W \subseteq \Ten_\F(m,k)$ degenerates to a module
	$U \subseteq \Ten_{\F[\veps]}(m,k)$ with order $q$, written $W \tleq_q U$ if every tensor in $W$
	degenerates to $U$ with order $q$.
\end{definition}

\begin{definition}
For any submodule $U \subseteq \Ten_{\F[\veps]}(m,k)$, we define 
$$
r(U) = \min\{r \ | U \subseteq \Ten_{\F[\veps]}(m,k)_{\leq r}\}.
$$
\end{definition}

The following corollary is straightforward.

\begin{corollary}
Suppose $U \subseteq \Ten_{\F[\veps]}(m,k)$ is an $F[\veps]$-submodule. Suppose $T \tleq_q U$ for some $q \in \Z_{\geq 1}$. Then 
$$
\brk(T) \leq r(U).
$$
\end{corollary}

We can now extend Lemma~\ref{basic2r} to the border rank setting:
\begin{lemma}\label{border-basic2}
Let $S \subseteq V$ be a spanning subset (of simples). Suppose $\phi:V \rightarrow \Ten(m,k)$ is a linear map and suppose $\phi(S) \tleq_q \sum_{\I} U_{\I}$ for some $\F[\veps]$-submodules $U_{\I} \subseteq \Ten_{\F[\veps]}(m,k)$. Then $\phi(V) \tleq_q \sum_{\I} U_{\I}$. In particular, for any $T \in V$, we have 
$$
\brk(\phi(T)) \leq \sum_{\I} r(U_{\I}).
$$
\end{lemma}

We define basic submodules $U$ for which we can upper bound $r(U)$.

\begin{definition}
A submodule of the form $U = U_1 \otimes U_2 \otimes \dots \otimes U_k \subseteq \Ten_{\F[\veps]}(m,k)$ is called a basic $F[\veps]$-submodule. 
\end{definition}

\begin{lemma}
Suppose $U = U_1 \otimes U_2 \otimes \dots \otimes U_k \subseteq \Ten_{\F[\veps]}(m,k)$ is basic $F[\veps]$-submodule. Then for any $1 \leq p \leq k$,
$$
r(U) \leq\prod_{i \neq p} \rk_{\F[\veps]}(U_i).
$$
\end{lemma}

\begin{definition} [Border potency]
Suppose $\phi:\Ten(n,d) \rightarrow \Ten(m,k)$ is a linear map, i.e., a $T_k$-rank method. Suppose $\brk(\phi(s)) \leq a$ for all $s \in S$. Then, for any $T \in \Ten(n,d)$, we have
$$
\brk(T) \geq \frac{\brk(\phi(T))}{a}.
$$
Thus, this is a method to prove lower bounds on border rank of tensors in $\Ten(n,d)$. Analogous to potency, we define border potency as
$$
\Bpot(\phi) = \frac{\max\{\brk(\phi(T))\  |\  T \in \Ten(n,d)\}}{a}.
$$
\end{definition}

\begin{theorem}\label{thm:border-TtoT}
	For any $T_k$-rank method $\phi : \Ten(n,d) \rightarrow \Ten(m,k)$, its border potency is
	$$ \Bpot(\phi) \leq k^d \cdot n^{\lfloor (k-1)d/k \rfloor}. $$
\end{theorem}
One can also prove the analogous result.

\begin{theorem}\label{thm:border-PtoT}
	For any $T_k$-rank method $\phi : P(n,d) \rightarrow \Ten(m, k)$, its border potency is
	$$ \Bpot(\phi) \leq B_{d,k} \cdot n^{\lfloor (k-1)d/k \rfloor}. $$
\end{theorem}

For simplicity, we shall only prove Theorem~\ref{thm:border-TtoT} in the special case where $k=3$. The proof parallels the proof of Theorem~\ref{TtoT}, with the additional complication of having to work with degenerations.

\begin{proof}[Proof of Theorem~\ref{thm:border-TtoT}, $k=3$ case]
	Let $S \subset \Ten_\F(n, d)$ be the set of tensors with rank $1$, parametrized by
	$\psi : (\F^n)^d \rightarrow \Ten_\F(n, d)$ given by 
	$(\bv_1, \ldots, \bv_d) \mapsto \bv_1 \otimes \cdots \otimes \bv_d$. Let 
	$\phi : \Ten_\F(n,d) \rightarrow \Ten_\F(m,3)$ be a linear map, i.e., a $T_k$-rank method and
	let $a = \max\{ \brk(\phi(s)) \mid s \in S\}.$
	
	Let $\bz = (\bz_1, \ldots, \bz_d)$ be a set of variables such that 
	$\bz_i = (z_{i1}, \ldots, z_{in})$ for each $i \in [d]$. Let $L = \phi \circ \psi$.
	Since $\brk(\phi(s)) \leq a$ for all $s \in S$, we have that $\brk(L(\beta)) \leq a$  for all 
	$\beta \in (\F^n)^d$. Equations for border rank are defined over $\F$
	 (see Corollary\ref{eqns.extn.field}), so we have that 
	the border rank of the symbolic tensor $L(\bz)$ is also $\leq a$. More precisely, 
	Corollary~\ref{eqns.extn.field} implies that $\brk_\K(L(\bz)) \leq a$, where border 
	rank is over the field $\K := \overline{\F(\bz)}$. We must go to the algebraic closure so that we can switch from the notion of (topological) border rank to algebraic border rank.
	 	
	Thus, by Theorem~\ref{thm:top-eq-alg-rk} over the field $\K$, we have that there exist tensors
	$T_1 \in \Ten_{\K[\veps]}(m,3)_{\leq a}$ and $T_2 \in \Ten_{\K[\veps]}(m, 3)$ and $q \in \N$ 
	such that 
	$$ \veps^{q-1} \cdot L(\bz) = T_1 + \veps^q \cdot T_2. $$
	
	Write $T_1 = \sum_{\ell=1}^a \bp_{\ell} \otimes \bq_{\ell} \otimes \br_{\ell}$, where 
	$\bp_{\ell}, \bq_{\ell}, \br_{\ell} \in \K[\veps]^m$ and denote the entries of $T_2$ by 
	$T_2(i,j,k) \in \K[\veps]$, where $i,j,k \in [m]$. 

	Since each entry of $\bp_{\ell}, \bq_{\ell}, \br_{\ell}, T_2$ is a
	polynomial in $\K[\veps]$, there exists $D \in \N$ such that we can write 
	$$ \bp_{\ell}(\bz, \veps) = \sum_{d=0}^D \bp_{\ell,d}(\bz) \cdot \veps^d,  $$ 
	similarly for $\bq_{\ell}$ and $\br_{\ell}$ and we can write
	$$ T_2(i,j,k)(\bz, \veps) = \sum_{d=0}^D T_2(i,j,k,d)(\bz) \cdot \veps^d, $$
	where $\bp_{\ell,d}(\bz), \bq_{\ell,d}(\bz), \br_{\ell,d}(\bz)$ are
	vectors in $\K^m$ and $T_2(i,j,k,d)(\bz) \in \K.$
	
	Let $C \subset \K$ be the set of all entries of 
	$\bp_{\ell,d}, \bq_{\ell,d}, \br_{\ell,d}$ and of all $T_2(i,j,k,d)$, for all
	ranges of $i,j,k,d, \ell$. 
	 $C$ is a finite set. 
	Therefore, Proposition~\ref{prop-ps} applies, and there exists a choice of 
	$\bc := (\bc_1, \ldots, \bc_d) \in (\F^n)^d$ such that all elements of $C$ have a power series 
	decomposition around the point $\bc$. Thus, this yields:
	
	\begin{equation}\label{eq:brk-symbolic} 
	\veps^{q-1} \cdot L(\bz+\bc) = \veps^q \cdot \hat{T}_2(\bz+\bc, \veps) \ + \ 
	\sum_{\ell=1}^a \hat{\bp}_{\ell}(\bz, \veps) \otimes \hat{\bq}_{\ell}(\bz, \veps) 
	\otimes \hat{\br}_{\ell}(\bz, \veps), 
	\end{equation}
	where $\hat{\bp}_{\ell}, \hat{\bq}_{\ell}, \hat{\br}_{\ell}, \hat{T}_2$ are 
	given by the power series decomposition around $\bo$. More precisely, they are given by:
	\begin{equation}\label{eq:b-pow-ser1} \hat{\bp}_{\ell}(\bz, \veps) = 
	\sum_{d=0}^D \sum_{\be \in \N^{dn}} \bp_{\ell, d, \be} \cdot \bz^\be \cdot \veps^d =
	\sum_{\be \in \N^{dn}} \tilde{\bp}_{\ell, \be}(\veps) \cdot \bz^\be,   
	\end{equation}
	similarly for $\bq_{\ell}$ and $\br_{\ell}$ and
	\begin{equation}\label{eq:b-pow-ser2} T_2(i,j,k)(\bz, \veps) = 
	\sum_{d=0}^D \sum_{\be \in \N^{dn}} T_2(i,j,k,d)_\be \cdot \bz^\be \cdot \veps^d =
	\sum_{\be \in \N^{dn}} \tilde{T}_2(i,j,k)_\be(\veps) \cdot \bz^\be, 
	\end{equation}
	where $\bp_{\ell,d,\be}, \bq_{\ell,d, \be}, \br_{\ell,d, \be} \in \F^m$ and 
	$T_2(i,j,k,d)_\be \in \F$ are the coefficients of the power series expansions, and
	$\tilde{\bp}_{\ell, \be}(\veps), \tilde{\bq}_{\ell, \be}(\veps), 
	\tilde{\br}_{\ell, \be}(\veps) \in \F[\veps]^m$
	and $\tilde{T}_2(i,j,k)_\be(\veps) \in \F[\veps]$ are simply the coefficients we obtain by
	grouping the elements of the power series with same monomial $\bz^\be$.

		
	Recall that $\SP(d,3)$ denotes the set of all partitions of $[d]$ into 
	$3$ sets. For $\I = (I_p,I_q, I_r) \in \SP(d,3)$ and $\ell \in [a]$, define
	$$ \cC^{\ell}_{\I} := \cC^{\ell}_{I_p} \otimes \cC^{\ell}_{I_q} \otimes \cC^{\ell}_{I_r}, $$
	where
	$$ \cC^{\ell}_{I_p} := \spa_{\F[\veps]}(\tilde{\bp}_{\ell, \be}(\veps) \mid 
	\smdeg(\be) = \bdelta_{I_p}) \subseteq \F[\veps]^m.$$
	Again, note here that $\bdelta_{I_p}$ denotes the indicator vector for the subset $I_p \subseteq [d]$. $\cC^{\ell}_{I_q}, \cC^{\ell}_{I_r}$ are analogously defined. 
	Since the entries of $L(\bz)$ are all set multi-homogenous of $\smdeg (1,1,\dots,1)$,  
	equations~\eqref{eq:brk-symbolic} and~\eqref{eq:b-pow-ser1} give us
	\begin{align*} 
	\veps^{q-1} \cdot L(\bz) 
	&= \veps^{q-1} \cdot L(\bz + \bc)_{(1,\ldots,1)} \\
	&= \veps^q \cdot \hat{T}_2(\bz+\bc, \veps)_{(1,\ldots,1)} \ + \ 
	\sum_{\ell=1}^a \left( \hat{\bp}_{\ell}(\bz, \veps) \otimes \hat{\bq}_{\ell}(\bz, \veps) 
	\otimes \hat{\br}_{\ell}(\bz, \veps) \right)_{(1,\ldots,1)} \\
	&= \veps^q \cdot \hat{T}_2(\bz+\bc, \veps)_{(1,\ldots,1)} \ + \ 
	\sum_{\ell=1}^a \sum_{(\be_p,\be_q,\be_r) \in \cJ} 
	\tilde{\bp}_{\ell, \be_p}(\veps) \otimes \tilde{\bq}_{\ell, \be_q}(\veps) 
	\otimes \tilde{\br}_{\ell, \be_r}(\veps) \bz^{\be_p + \be_q + \be_r}, 
	\end{align*}  

	where $\cJ = \{ (\be_p, \be_q, \be_r) \mid 
	\smdeg(\be_p) + \smdeg(\be_q) + \smdeg(\be_r) = (1, \ldots, 1)  \}$, that is,
	the set of monomials in $\bz$ such that their set mutli-degree adds to the $(1,\ldots, 1)$
	vector.	We will now prove that 
	$$\phi(S) \tleq_q \sum_{\ell=1}^a \sum_{\I \in \SP(d,3)} \cC^{\ell}_{\I}.$$
	
	Note that for any $s \in S$, we have $s = \psi(\balpha)$ for some $\balpha \in (\F^n)^d$. 
	So, from the above equation, we get
	
	$$
	\veps^{q-1} \phi(s) = \veps^q \cdot \hat{T}_2(\balpha +\bc, \veps)_{(1,\ldots,1)} \ + \ 
	\sum_{\ell=1}^a \sum_{(\be_p,\be_q,\be_r) \in \cJ} 
	\tilde{\bp}_{\ell, \be_p}(\veps) \otimes \tilde{\bq}_{\ell, \be_q}(\veps) 
	\otimes \tilde{\br}_{\ell, \be_r}(\veps) \balpha^{\be_p + \be_q + \be_r} 
	$$ 
	
	Thus, it suffices to show that $$
	\sum_{\ell=1}^a \sum_{(\be_p,\be_q,\be_r) \in \cJ} 
	\tilde{\bp}_{\ell, \be_p}(\veps) \otimes \tilde{\bq}_{\ell, \be_q}(\veps) 
	\otimes \tilde{\br}_{\ell, \be_r}(\veps) \balpha^{\be_p + \be_q + \be_r} \in \sum_{\ell=1}^a \sum_{\I \in \SP(d,3)} \cC^{\ell}_{\I}.
	$$
	
	Pick any term $t = \tilde{\bp}_{\ell, \be_p}(\veps) \otimes \tilde{\bq}_{\ell, \be_q}(\veps) 
	\otimes \tilde{\br}_{\ell, \be_r}(\veps) \balpha^{\be_p + \be_q + \be_r}$. Let $I_p$ be the support subset of $\smdeg(\be_p)$, i.e., the subset of positions with non-zero entries. Define $I_q,I_r$ similarly, and let $\I = (I_p,I_q,I_r) \in \SP(d,3)$. Then $t \in \cC^{\ell}_{\I}$.
	
	Hence, we have shown that that $\phi(S) \tleq_q \sum_{\ell=1}^a \sum_{\I \in \SP(d,3)} \cC^{\ell}_{\I}.$ Applying Lemma~\ref{border-basic2}, we deduce that for all $T \in V$, 
	
$$
\brk(T) \leq  \sum_{\ell=1}^a \sum_{\I \in \SP(d,3)} r(\cC^{\ell}_{\I}).
$$

So, all that is left is to upper bound the right hand side. But this is precisely the same 
calculation from Lemma~\ref{lem:tensor-count}, giving us the required upper bound on border potency. 
\end{proof}

\section{Generalizations, Cactus rank and improvements to rank methods}\label{sec:gen-cac-imp}
The aim of this section is two fold. First, there is an alternative approach to establishing barriers for rank methods using the notion of cactus rank. By infusing our techniques with a little trick, we will show that both approaches establish barriers by counting the same things. Despite this, there seems to be no straightforward connection between the two approaches. We want to point out in particular that the barriers to the generalized rank methods that we prove in this paper have no analogue in the cactus rank approach. Second, we want to extend the barriers for matrix-rank methods to the setting of set multi-homogenous rank, which is a generalization of both Waring and tensor rank.

We first make a simple observation:

\begin{lemma}
Suppose $T \subseteq S$ are two spanning subsets of $V$. Then for any matrix rank method $\phi: V \rightarrow \M_{k,l}$ its potency for computing lower bounds on $S$-rank is less than or equal to its potency for computing lower bounds for $T$-rank.
\end{lemma}

So, proving upper bounds for potency of matrix-rank methods for $T$-rank will automatically prove upper bounds for potency of matrix-rank methods for $S$-rank. The proof of upper bounds for potency really only depends on the parametrization of $S$. Roughly speaking, since $T$ is smaller, we might be able to get a smaller parametrization which could help prove sharper bounds. Let us exhibit this explicitly in the case of Waring rank.

\begin{lemma}
For any matrix-rank method $\phi: P(n+1,d) \rightarrow \M_{k,l}$, we have
$$
\Pot(\phi) \leq Y_{n,d} + Z_{n,d}.
$$
\end{lemma}

First, note that this is indeed stronger than the statement of Theorem~\ref{barr-egow-war} in the introduction because we are considering degree $d$ polynomials in $n+1$ variables (as opposed to $n$ variables). 

\begin{proof}
Consider the subset $T = \{(a_1x_1 + \dots +a_nx_n + x_{n+1})^d \ | a_i \in \F\} \subseteq S = \{\ell^d\ \ell \in P(n+1,d)\}$. We leave it to the reader to check that $T$ is also a spanning subset\footnote{It suffices to check that $\spa(T)$ contains $S$.}. Now, observe that $\pi: \F^n \rightarrow P(n+1,d)$ given by $(a_1,\dots,a_n) \mapsto (a_1x_1 + \dots +a_nx_n + x_{n+1})^d$ parametrizes $T$. This parametrization requires $n$ variables. One should note that while $\psi$ was homogenous map of degree $d$, $\pi$ is not. However, the homogenous components of the map $\pi$ are all of degree $\leq d$. This is sufficient. By replacing $\psi$ by $\pi$ in the proof of Theorem~\ref{barr-egow-war}, we get the required upper bounds on potency of $T$-rank, and hence upper bounds on potency of $S$-rank.
\end{proof}

Let us now define the set multi-homogeneous rank.

\begin{definition} \label{smh}
Let $\bn = (n_1,n_2,\dots,n_k)$ and $\bd = (d_1,\dots,d_k)$, and let $d = d_1 + \dots + d_k$. Let 
$$
V(\bn,\bd) = P(n_1 + 1, d_1) \otimes P(n_2 + 1,d_2) \otimes \dots P(n_k+1,d_k).
$$
Let $S(\bn,\bd) = \{\ell_1^{d_1} \otimes \ell_2^{d_2} \otimes \dots \otimes \ell_k^{d_k} \ | \ \ell_i \in P(n_i + 1)\} \subseteq V(\bn,\bd)$ be a set of simples. Then for $v \in V$, we define 
$$
\rk_{\bn,\bd}(v) := \rk_{S(\bn,\bd)} (v).
$$
\end{definition} 

Note that each $\ell_i$ is a linear form in $n_i + 1$ variables rather than $n_i$ variables. In particular, tensor rank for tensors in $\Ten(n,d)$ is $\rk_{(n-1,n-1,\dots,n-1), (1,1,\dots,1)}$ and Waring rank for degree $d$ homogeneous polynomials in $n$ variables is $\rk_{n - 1,d}$. $S(\bn,\bd)$ is a subvariety of $V(\bn,\bd)$ and is sometimes called the Segre-Veronese variety.\footnote{To be precise it is the affine cone over the Segre-Veronese variety.}

Using the same proof as Theorem~\ref{barr-egow-war} and Theorem~\ref{barr-egow-ten}, along with the additional improvement given by the lemma above, we get:

\begin{theorem} \label{theo:improv}
Let $\bz = (\bz_1,\dots,\bz_k)$ where each $\bz_i = (\bz_{i1},\dots,\bz_{in_i})$ denote a set of variables, and define set multi-grading as before. For any rank method $\phi: V(\bn,\bd) \rightarrow \M_{p,q}$, its potency is upper bounded by
$$
\Pot(\phi) \leq Y_{\bn,\bd} + Z_{\bn,\bd},
$$
where
$$
Y_{\bn,\bd} = \text{number of monomials in $\bz$ of $\smdeg \preceq \bd$ and total degree $\leq \lfloor d/2 \rfloor$, and}
$$

$$
Z_{\bn,\bd} = \text{number of monomials in $\bz$ of $\smdeg \preceq \bd$ and total degree $\leq d  - (\lfloor d/2 \rfloor + 1)$}.
$$
\end{theorem}

We omit the details. The number $Y_{\bn,\bd} + Z_{\bn,\bd}$ is the upper bound on the cactus rank obtained in \cite{G16ma}. An explicit upper bound for $Y_{\bn,\bd} + Z_{\bn,\bd}$ can be found on \cite[Page~18]{G16ma}. Let us state the bounds one obtains for the potency of matrix-rank methods for tensor rank and Waring rank with these improvements.

\begin{corollary}
Specializing the above result, we get the following:
\begin{itemize}
\setlength\itemsep{1em}
\item An upper bound of $N(n+1,d)$ on the potency of rank methods for Waring rank of degree $d$ homogeneous polynomials in $n+1$ variables, where
$$
N(n+1,d) = \begin{cases} 2{n+k \choose k} & \text{when } d = 2k+1, \\
 {n +k \choose k} + {n+k+1 \choose k+1} & \text{when } d = 2k+2.
  \end{cases}
$$
This is equal to the cactus rank bound obtained in \cite[Theorem~3]{BR13}.

\item An upper bound of $M(n+1,d)$ for the potency of rank methods for tensor rank for tensors in $\Ten(n+1,d)$, where
$$
M(n+1,d) = \begin{cases} 2 \left(1 + dn + {d \choose 2}n^2 + \dots  + {d \choose \lfloor d/2 \rfloor} n^{\lfloor d/2 \rfloor} \right) & \text{if $d$ is odd,} \vspace{10pt} \\

2 \left(1 + dn + {d \choose 2}n^2 + \dots  + {d \choose \lfloor d/2 \rfloor - 1} n^{\lfloor d/2\rfloor -1} \right) + {d \choose \lfloor d/2 \rfloor} n^{\lfloor d/2 \rfloor} & \text{if $d$ is even.} 
\end{cases}
$$
This is equal to the cactus rank bound obtained in \cite[Example~6.3]{G16ma}.

\item An upper bound of $2n_1 + 2n_2 + 2n_3 - 4$ for the potency of rank methods for tensor rank of tensors in $\F^{n_1} \otimes \F^{n_2} \otimes \F^{n_3}$. This is equal to the cactus rank bound obtained by~\cite{B18p}. It also follows from the results in \cite{G16ma}.
\end{itemize}

\end{corollary}

\section{Elusive functions, and the potential for symbolic methods in lower bounds} \label{sec:elusive}
Our aim in this section is to put forth a symbolic perspective on the notion of {\em elusive functions}, and expose some of the advantages in doing so. Elusive functions were defined by Raz in~\cite{Raz-elusive}, where his main result is that explicit\footnote{This notion is formally defined in the paper, but is essentially the usual notion: there is a polynomial-time algorithm computing the coefficient of each monomial.} elusive functions (for suitable parameters) will imply super-polynomial lower bounds in arithmetic complexity, thus separating VP from VNP.  Let us begin by defining elusive functions.

We say that a polynomial map $M = (M_1,\dots,M_m)$ is of degree $d$, if each $M_i$ is a polynomial function of degree at most $d$ (not necessarily homogeneous). 

\begin{definition} [$(r,d)$-elusive]
We say a polynomial map $L: \F^n \rightarrow \F^m$ is $(r,d)$-elusive if for every polynomial mapping $M:\F^r \rightarrow \F^m$ of degree $d$, $\im(L) \not\subset \im(M)$.  
\end{definition}

The striking feature of this definition in the context of our paper is that it cares about inclusion of images of polynomial maps. This is a ``numeric'' statement. But recall that the {\em hypothesis} of our ``numeric to symbolic'' transfer  (Theorem~\ref{main-ag}) is also a similar ``numeric'' condition on the inclusion of images of polynomial maps. Its conclusion however is ``symbolic'', and so we can potentially use this conclusion to prove elusiveness!

In this section we will actually use only Proposition~\ref{prop:gensymb}, the `first half' of the Theorem~\ref{main-ag} (see discussion in Section~\ref{sec:new ideas}).
Using it, we can give a symbolic point of view of elusiveness (and with it, non-elusiveness). Before doing so, we need a definition.

\begin{definition} [degree $d$-span]
Let $\bz = (z_1,\dots,z_n)$ denote indeterminates. For \\
$p_1(\bz), \dots, p_r(\bz) \in \overline{F(\bz)}$, we define its degree $d$-span 
$$
d\mbox{-}\spa (p_1(\bz),\dots,p_r(\bz)) = \spa_{\F} \left ( p_1(\bz)^{e_1} p_2(\bz)^{e_2} \dots p_r(\bz)^{e_r} \ : \ \sum_i e_i \leq d \right).
$$

In other words, the $\F$-linear span of all the monomials in the $p_i(\bz)$'s of degree at most $d$.
\end{definition}

\begin{lemma} \label{elus-symb}
Let $\bz = (z_1,\dots,z_n)$ denote indeterminates. If the polynomial map $L: \F^n \rightarrow \F^m$ is not $(r,d)$-elusive, then there exist $p_1(\bz),p_2(\bz),\dots,p_r(\bz) \in \overline{\F(\bz)}$ such that for each $i$,  $L_i \in d\mbox{-}\spa (p_1(\bz),\dots,p_r(\bz))$.
\end{lemma}

\begin{proof}
Suppose $L$ is not $(r,d)$ elusive, then $\exists$ degree $d$ polynomial map $M:\F^r \rightarrow \F^m$ such that $\im(L) \subseteq \im(M)$. This means, by Proposition~\ref{prop:gensymb} that there exists $p_1(\bz),p_2(\bz),\dots,p_r(\bz) \in  \overline{\F(\bz)}$ such that $L(z_1,\dots,z_n) = M(p_1(\bz),p_2(\bz),\dots,p_r(\bz))$. Since, each $M_i$ is a degree $d$ polynomial (sum of monomials), we get that 
$$
L_i(\bz) = M_i(p_1(\bz),p_2(\bz),\dots,p_r(\bz)) \in d\mbox{-}\spa (p_1(\bz),\dots,p_r(\bz)).
$$
\end{proof}

\begin{remark}
One can go further and directly apply Theorem~\ref{main-ag} to get a similar looking statement where you replace the algebraic functions $p_i(\bz)$, which live in the algebraic closure  $\overline{\F(\bz)}$, with power series (after a suitable shift), defined over the base field $\F$. This may be even more powerful. 
\end{remark}

The first example of an elusive function is the well known moment curve, provided by Raz~\cite{Raz-elusive} as a motivating example:

\begin{proposition} \label{elus-mom}
The map $L: \F \rightarrow \F^m$ given by $x \mapsto (x,x^2,\dots,x^m)$ is $(m-1,1)$-elusive.
\end{proposition}

Using the definition of an elusive function, one can see that the above proposition simply asserts that the moment curve is not contained in any affine hyperplane.  The most straightforward proof of this assertion is based on the invertibility of the Vandermonde matrix, namely the linear independence of any $m$ distinct vectors in the image of the moment curve. However, from the symbolic interpretation, the above lemma essentially becomes a consequence of a linear independence of the $m$ monomials in the description of the moment curve, as we describe below.

\begin{proof} [Proof of Lemma~\ref{elus-mom}]
Suppose $L$ is not $(m-1,1)$-elusive. Then by Lemma~\ref{elus-symb}, we have $p_1(z),\dots,p_{m-1}(z) \in \overline{\F(z)}$ such that 
$$
L_i(z) = z^i \in 1\mbox{-}\spa(p_1(z),\dots,p_{m-1}(z)) = \spa_{\F}(1,p_1(z),\dots,p_{m-1}(z))
$$
for $1 \leq i \leq m$, where $z$ is an indeterminate. But this means that $\spa_{\F}(z,z^2,\dots,z^m)  \subseteq \spa_{\F}(1,p_1(z),\dots,p_{m-1}(z))$. The former is an $m$-dimensional linear space, by linear independence of the $z^i$, and the latter is at most $m$-dimensional (as it is a span of $m$ elements). Hence, $\spa_{\F}(z,z^2,\dots,z^m)  = \spa_{\F}(1,p_1(z),\dots,p_{m-1}(z))$. But $1 \notin \spa_{\F}(z,z^2,\dots,z^m)$, which is a contradiction. Thus, $L$ must be $(m-1,1)$-elusive.
\end{proof}

It is of course not surprising that the linear independence of monomials is very much related to the Vandermonde matrix. The numeric to symbolic transfer simply recasts the numeric "invertibility of Vandermonde matrix" as a symbolic  "linear independence of monomials". While the invertibility of Vandermonde matrix is well known, it is not completely obvious. On the other hand, the linear independence of monomials is completely straightforward from a symbolic perspective. In some sense, we let the (non-trivial!) numeric to symbolic transfer statement do the `heavy-lifting'. Indeed, notice that the exact same proof above actually yields the following much more general proposition (which again, can be obtained ``numerically'', but not with such simplicity).

\begin{proposition} \label{elus-indep}
Any polynomial map $L: \F \rightarrow \F^m$ given by $x \mapsto (p_1(x), p_2(x), \dots ,p_m(x))$, for which the polynomials $\{p_i\} \cup \{1\}$ are linearly independent is $(m-1,1)$-elusive.
\end{proposition}

Elusive functions for degree $d=1$ cannot yield arithmetic lower bounds. Surprisingly, Raz (\cite{Raz-elusive}) proves that already for degree $d=2$, explicit elusive functions of appropriate parameters can separate VP from VNP! More specifically, he proves


\begin{theorem}\cite{Raz-elusive}
Any explicit polynomial map $L: \F^n \rightarrow \F^m$ (of degree at most ${\rm poly}(n)$) which is $(r,2)$-elusive, with $m \geq n^{\omega(1)}$ and $r \geq m^{0.9}$, implies that {\rm VP} $\neq$ {\rm VNP}.
\end{theorem}

This beautiful avenue to proving superpolynomial lower bounds is a great challenge to our techniques, and no progress we know of was made since that paper came out. Here we will attempt to handle a very toy version of it using our numeric to symbolic transfer. While a toy, unlike the moment curve above, we don't know of a way to probe that toy result ``numerically''.

Indeed, one virtue of the symbolic perspective is that it provides several relaxations of the notion of elusiveness. Establishing elusiveness of a function is really hard (not surprisingly), and these relaxations provide intermediate problems that could aid our understanding. 

The map we consider is again a curve, $L:\F \rightarrow \F^{m+1}$ given by $x \mapsto (x,x^3,x^9,\dots,x^{3^m})$,  of monomials with exponentially growing degrees. It is a toy, namely very restricted example in two essential ways. First, as it happens, to match it with Ran's parameters, to prove a lower bound using a curve one would need the monomials degrees to grow much slower.\footnote{In that case one could use extra variables, and encode the curve $L$ as a polynomial map $L': \F^n \rightarrow \F^m$ satisfying the condition $m \geq n^{\omega(1)}$.} Second, we will not be able to rule out any map $M$ as in the definition of elusiveness, but only ones defined by monomials. In this simple case we can actually get $r=m-1$, as for the moment curve. We do not know how to extend it to arbitrary polynomials, let alone algebraic functions. Indeed, extending this result even to ``monomials'' with negative exponents, seems like a challenging problem.





\begin{proposition}
$L:\F \rightarrow \F^{m+1}$ that maps $x \mapsto (x,x^3,x^9,\dots,x^{3^m})$. Let $z$ be an indeterminate. Then for any choice of monomials $z^{e_1},z^{e_2},\dots,z^{e_{m-1}}$ with $e_i \in \Q_{\geq 0}$, there is some $i$ such that $L_i(z) \notin 2\mbox{-}\spa(z^{e_1},z^{e_2},\dots,z^{e_{m-1}})$.
\end{proposition}

\begin{proof}
 Let $z$ be an indeterminate. Suppose $\exists$ monomials $z^{e_1},z^{e_2},\dots,z^{e_{m-1}}$ (with $e_i \in \Q_{\geq 0}$) such that for each $i$, $L_i(z) = z^{3^i} \in 2\mbox{-}\spa(z^{e_1},z^{e_2},\dots,z^{e_{m-1}})$. The only monomials in the $2\mbox{-}\spa(z^{e_1},z^{e_2},\dots,z^{e_{m-1}})$ are of the form $1$ or $z^{e_i}$ or $z^{e_i} z^{e_j} = z^{e_i + e_j}$. Thus, we must have 
 $$
 \{1,3,3^2,\dots,3^m\} \subseteq W:= \{0,e_i, e_i+e_j\ |\ 1 \leq i,j \leq m\}.
 $$

The proof will be by induction on $m$. Without loss of generality, we can assume that the $e_i$'s are in increasing order. Suppose $e_1 > 1$, then every $e_i$ and $e_i + e_j$ are all greater than $1$. But then $1 \notin W$, which is a contradiction. Thus, we must have $e_1 \leq 1$. 

Next, suppose $e_2 > 3$, then the only elements in $W$ that are $\leq 3$ are $0,e_1,2e_1$. But since $e_1 \leq 1$, we have that $0,e_1,2e_1 < 3$. This is a contradiction, so we must have $e_2 \leq 3$. Continuing by induction, we must have $e_i \leq 3^{i-1}$ for all $1 \leq i \leq m$. But now, the largest number in $W$ is $2e_{m} \leq 2 \cdot 3^{m-1}$, which is smaller than $3^m$. This means $3^m \notin W$, which is a contradiction.
\end{proof}

One sees immediately the (symbolic) notion of degree that is crucially used in proving this result. Although, we only illustrated its use in a toy case, the notion of degree could still be important (along with other ideas) in studying elusiveness. There needs to be more work done to understand what features (such as degree) the symbolic view point offers, and what these features are worth in our understanding of elusive functions.


To summarize, proving lower bounds via elusive functions is an intriguing strategy, and the difficulties and possibilities of this approach need to be explored. The symbolic view point (that results from applying numeric to symbolic transfer) gives a fresh perspective. We find that this approach needs further analysis, and could lead to new and exciting results.

\subsection*{Acknowledgements}
We would like to thank Rankeya Datta, Christian Ikenmeyer, Camillo De Lellis, Daniel Litt, Neeraj Kayal, Nitin Saxena and Akash Sengupta for helpful discussions.

\appendix
\section{Equations for border rank and field extensions}

In this section, it will suffice to assume that $\F$ is an infinite field. One feature of the notion of tensor rank is that the canonical parametrization of simples is essentially 
independent of the field. A more formal way of phrasing this is to say that the canonical parametrization is 
compatible with base change. Indeed, we have the parametrization 
$\psi_{\F}: (\F^n)^{\times d} \rightarrow \Ten_{\F}(n,d)$. Observe that for any extension field 
$\K$ of $\F$, we have $\psi_{\F} \otimes_{\F} \K = \psi_{\K}$. This is what we mean by compatibility with base 
change. The compatibility of this parametrization with base change is the core reason for the fact that 
equations for border rank over $\K$ are actually defined over $\F$.

Let $I_{\F,r} \subseteq \F[\Ten_{\F}(n,d)]$ denote the ideal of polynomials that vanish on 
$\Ten_{\F}(n,d)_{\leq r}$. Note that the zero set of $I_{\F,r}$ is the set of all tensors of border rank 
$\leq r$, which we will denote by $\overline{\Ten_{\F}(n,d)_{\leq r}}$. 

\begin{proposition} \label{field.ind}
Suppose $\K$ is an extension field of $\F$. Then $I_{\F,r} \otimes_{\F} \K = I_{\K,r}$. In particular, 
generators for the ideal $I_{\F,r}$ are also generators for the ideal $I_{\K,r}$. In other words, we
have that the equations for border rank over $\Ten_{\K}(n,d)$ are defined by polynomials over
$\F[\Ten_{\F}(n,d)]$.
\end{proposition}

Before we prove the proposition, let us observe the following consequence, which is required in 
Section~\ref{sec:border-rank}. Let $\bz = (z_1,\dots,z_m)$ denote a vector of indeterminates, and let 
$\F(\bz) = \F(z_1,\dots,z_m)$ denote the function field in $m$ variables. We call 
$L(\bz) \in \Ten_{\F(\bz)}(n,d)$ a polynomial tensor, if each of the $n^d$ entries consist of polynomials. 
In more precise terms, $L(\bz) \in \Ten_{\F[\bz]}(n,d)$, where $\F[\bz] = \F[z_1,\dots,z_m]$ denotes 
the polynomial ring.

\begin{corollary} \label{eqns.extn.field}
Let $L(\bz) \in \Ten_{\F(\bz)}(n,d)$ be a polynomial tensor. Suppose $\brk_{\F}(L(\bbeta)) \leq r$ for all 
$\bbeta \in \F^m$. Then $\brk_{\F(\bz)}(L(\bz)) \leq r$, and further 
$\brk_{\overline{\F(\bz)}} (L(\bz)) \leq r$. 
\end{corollary}

\begin{proof}
Border rank can at best drop when we consider the same tensor over an extension field. So, it suffices to 
prove that $\brk_{\F(\bz)} (L(\bz)) \leq r$. Let $(f_1,\dots,f_t)$ denote the generators for the ideal 
$I_{\F,r}$. Then by Proposition~\ref{field.ind}, these are also generators for $I_{\F(\bz),r}$. Hence, 
it suffices to show that $f_i(L(\bz))= 0$ for all $i$. Observe that $f_i(L(\bz))$ is a polynomial in 
$\bz$ with coefficients in $\F$, let us call this polynomial $p_i(\bz)$. Now, we know that 
$\brk_\F(L(\bbeta)) \leq r$ for all $\bbeta \in \F^m$. This means that $f_i(L(\bbeta)) = p_i(\bbeta) = 0$ 
for all $\bbeta \in \F^m$. Since $\F$ is an infinite field, it implies that $p_i(\bz)$ is the identically 
zero polynomial, i.e., we have shown $f_i(L(\bz)) = 0$ as required.
\end{proof}

We will now give a proof of Proposition~\ref{field.ind}. 
We thank Christian Ikenmeyer for telling us this proof.

\begin{proof}[Proof of Proposition~\ref{field.ind}]
If we take a matrix $M$ with entries over $\F$, we can also interpret it as a matrix over $\K$. The kernel 
of this matrix is compatible with base change from $\F$ to $\K$. In the language of linear maps, this 
means that if we have a linear map $M: V \rightarrow W$ of $\F$-vector spaces, then 
${\rm Ker}(M) \otimes_{\F} \K = {\rm Ker}(M \otimes_{\F} \K)$. Hence, to see that the ideal of 
polynomials that vanishes on all tensors of tensor rank $\leq r$ is compatible with base change, 
we will describe this as the kernel of a linear map. This linear map will be defined over $\F$, so by 
the above observation, we get the required result.

Recall that $(\F^n)^{\times d} = \F^{nd}$ parametrizes rank $1$ tensors. This gives an obvious 
parametrization of rank $\leq r$ tensors, which we denote by $\psi: \F^{ndr} \rightarrow \Ten_{\F}(n,d)$. 
Let $\bz = (z_1,\dots,z_{ndr})$ denote indeterminates. Consider the linear map 
$M: \F[\Ten_{\F}(n,d)] \rightarrow \F[\bz]$ given by $p \mapsto p(\psi(\bz))$. If $p(\psi(\bz)) = 0$, then $p(\psi(\bbeta))= 0$ for all $\bbeta \in \F^{ndr}$. Since $\psi$ parametrizes $\Ten_{\F}(n,d)_{\leq r}$, this means that $p$ vanishes on $\Ten_{\F}(n,d)_{\leq r}$. Conversely, if $p$ vanishes on $\Ten_{\F}(n,d)_{\leq r}$, then $p(\psi(\bbeta))= 0$ for all $\bbeta \in \F^{ndr}$. Since $\F$ is infinite, this means that $p(\psi(\bz)) = 0$. To summarize, 
$I_{\F,r}$ is the kernel of this linear map $M$.

Note that we consider this as linear map between $\F$ 
vector spaces. We note that the vector spaces are infinite dimensional, but this doesn't become an issue. By a similar argument, $I_{\K,r}$ is the kernel of  $M \otimes_{\F} \K$. Hence, by the above discussion on base change compatibility of kernels of linear maps, we have 
$I_{\K,r} = I_{\F,r} \otimes_{\F} \K$ as required.
\end{proof}

\bibliographystyle{alpha}
\bibliography{refs}

%
%
%
%
%
%
%
%
%
%
%
%

\end{document}